\documentclass[11pt]{article}
\usepackage{amssymb,amsmath,amsthm,amsfonts}
\usepackage{graphicx}
\usepackage{subfigure}
\usepackage[usenames, dvipsnames]{color}
\usepackage{mathrsfs}
\usepackage[CJKbookmarks, colorlinks, bookmarksnumbered=true,pdfstartview=FitH,linkcolor=black]{hyperref}

\def\disp{\displaystyle}

\def\crr{\cr\noalign{\vskip2mm}}

\def\dref#1{(\ref{#1})}

\theoremstyle{plain}
\newtheorem{theorem}{Theorem}[section]
\newtheorem{lemma}{Lemma}[section]

\newtheorem{assumption}{Assumption}[section]

\numberwithin{equation}{section}

\theoremstyle{definition}
\newtheorem{example}{Example}[section]

\newcommand{\R}{{\mathbb R}}

\def\A{\mathcal{A}}

\begin{document}
\title{{\bf    Boundary  Stabilization and Observation of an   Unstable
Heat Equation in a General  Multi-dimensional  Domain  \footnote{\small This work is
   supported  by
the National Natural Science Foundation of China (61873153,11901365).}} }
\author{Hongyinping Feng\footnote{\small
 Corresponding author.  Email: fhyp@sxu.edu.cn.   },\ \   Pei-Hua Lang\ \   and\  \ Jiankang Liu
\\ {\it  School of Mathematical Sciences}\\
{\it Shanxi University,  Taiyuan, Shanxi, 030006, P.R.
China}\\
\\
}
\date{}
\maketitle
\allowdisplaybreaks[4]
\begin{abstract}
\normalsize
In this paper, we consider the  exponential stabilization and observation of
an     unstable heat  equation in a general  multi-dimensional domain  by combining the   finite-dimensional spectral truncation technique and  the
  recently developed dynamics compensation approach.
  In contrast to
 the unstable  one-dimensional partial differential equation (PDE),  such as the transport equation,  wave equation and the heat equation, that can be treated by the well-known PDE
  backstepping method,
  stabilization of   unstable PDE  in  a general multi-dimensional domain is still a challenging   problem. We treat the stabilization and observation problems separately.
A dynamical state feedback law is proposed  firstly to stabilize the unstable  heat equation
exponentially and then a state  observer is designed via a boundary measurement.
 Both the stability of the closed-loop system and the well-posedness of the observer are
  proved. Some of the theoretical results are   validated  by the numerical simulations.

\end{abstract}
\vspace{0.3cm}
\vspace{0.3cm}
\section{Introduction}\label{S1}

Since the backstepping approach was first introduced into the
systems described by the partial differential equations (PDEs) in
\cite{Liuw},  \cite{Andrey1},  and \cite{Andrey2},  the landscape of
one-dimensional PDEs control has completely changed.  This can be
seen from its success in stabilizing  the  unstable \cite{Andrey3,
Andrey4} or even anti-stable wave systems \cite{MK}, which are
almost formidable by other approaches. However, the backstepping
approach seems only effective for   one-dimensional PDEs
and
  there is a
formidable obstacle  to  applying this  approach to multi-dimensional PDEs except for
some    domains
of specific geometry.

Very recently, a new approach, called  dynamics compensation approach, has been developed to
 cope with the actuator and sensor dynamics compensations in the abstract framework (\cite{FPart1}, \cite{FPart2}).  Interestingly,   this
  approach can also be used to the control of unstable PDEs.
 As an example, the  stabilization  and observation of an unstable one-dimensional
heat equation have been  considered in \cite{FPart1} and  \cite{FPart2}, respectively.
In addition to the existing backstepping method,
 the dynamics compensation approach  gives a new and
  completely different  way to cope with the unstable PDEs.
  In this paper, we will apply this new approach to
     the multi-dimensional unstable heat equation in a general domain,  which is
     commonly recognized as  a longstanding puzzled problem.



%
%

Suppose that  $\Omega\subset \R^n (n \geq2)$   is  a bounded domain with $C^2$-boundary $\Gamma
  $,
    $ \Gamma_1 $ is a non-empty connected open set in $ \Gamma$,  $ \Gamma_0=\Gamma\setminus\overline{\Gamma_1}$ and
    $\Gamma_0\neq \emptyset$.
    Let $\nu$ be the unit
outward normal vector of $\Gamma_1$ and let $\Delta$ be the usual  Laplacian which is defined by
  \begin{equation}\label{2021162039}
  \Delta f= \sum_{i=1}^n \frac{  \partial^2 f }{  \partial x_i^2},\  \ \forall\ f\in H^2(\Omega).
\end{equation}
We consider the following  heat equation:
\begin{equation}\label{2020981903}
\left\{\begin{array}{l}
 w_{t}(x,t)= \Delta w (x,t)+\mu w(x,t),\ \   x\in \Omega, \ t> 0,\crr
 w (x,t)=0,\ \ \ x\in\Gamma_0,\ \ \disp \frac{\partial w(x,t)}{\partial \nu}=u(x,t),\ \ x\in\Gamma_1,\ \  t\geq0,\crr
 y(x,t)=w(x,t),\ \ x\in\Gamma_1,\ \ t\geq0,
\end{array}\right.
\end{equation}
where $w(\cdot,t)$ is the state,   $\mu >0$,   $u$ is the control and
$y$ is the output. System \dref{2020981903}
  is a general  heat equation with interior
convection.
In physics and engineering contexts, it
 describes the flow of heat in a homogeneous and isotropic medium, with $w(x, t)$ being the temperature at the point $x$ and time $t$.  The  more detailed   physical
 interpretation of the
  heat equation can be found in  \cite{Heatbook}.

  By a simple computation,    we can see  that there are some
   eigenvalues of the   open-loop system \dref{2020981903} (with $u(\cdot,t)\equiv0$)
in the right-half plane provided $\mu$  is sufficiently large. This shows that
the  open-loop system \dref{2020981903} is  unstable for large $\mu$. The
 lower-order term $\mu w(\cdot,t) $ of \dref{2020981903} is  usually referred to as source term or unstable term.
Heat equations with unstable term or source term have
  been extensively studied by the method of PDE  backstepping.
 Examples can be found in   \cite{Pisano}, \cite{Meu}, \cite{Backstepping1}, \cite{Backstepping2}, and \cite{Andrey3},
 to name just a few. The  PDE  backstepping method  is powerful and is still valid to other one-dimensional distributed parameter
systems  such as the   wave equation \cite{MK}, Schr\"{o}dinger  equation \cite{GuoLiu2104},
 the first order hyperbolic equation \cite{Kristic2008hyperbolic} as well as
some special
Euler-Bernoulli beam \cite{Andrey5}.
However,  the  application of backstepping method  seems to stop
  in front of  unstable   PDEs in the general multi-dimensional domain.
 There  still  exist   formidable obstacles  to applying  this approach to general multi-dimensional PDEs.


In this paper, we  combine the newly developed dynamics compensation approach    \cite{FPart1,FPart2} and the finite-dimensional spectral truncation
technique \cite{CoronTrelat2004SICON,PrieurandTrelat2019TAC}  to cope with the unstable system \dref{2020981903}.
The control objective is to stabilize the system exponentially  by virtue of the measurement output.
 Owing to   the separation principle of
the linear systems, the output feedback will be available   once
we address the following two  problems: (i), stabilize  system \dref{2020981903} by a full state feedback; (ii), design a  state observer in terms of the measurement output.
We will consider these two problems separately.

We consider  system \dref{2020981903} in the state space
  $  L^2(\Omega)$.
   Let
  \begin{equation}\label{20209101543}
  \left.\begin{array}{l}
\disp A f =\Delta f,\  \ \forall\ f\in D(A)
  =\left\{f\ |\; f\in H^2(\Omega)\cap H_{\Gamma_0}^1(\Omega),
\frac{\partial f}{\partial\nu}\big|_{\Gamma_1}=0\right\},
\end{array}\right.
\end{equation}
 where $H_{\Gamma_0}^1(\Omega)=\{f\in H^1(\Omega)|\; f=0
\hbox{ on }\Gamma_0\}$.
 Then $A$
   generates an exponentially
stable analytic semigroup on $L^2(\Omega)$.
 It is well known (e.g. \cite[p.668]{LT3})  that
 $D((-A)^{1/2})=H_{\Gamma_0}^1(\Omega)$ and $(-A)^{1/2}$ is a canonical
isomorphism from   $H^1_{\Gamma_0}(\Omega)$ onto  $L^2(\Omega)$. Moreover, the following
Gelfand triple compact inclusions  are valid:
\begin{equation}\label{4.3}
H^1_{\Gamma_0}(\Omega)=D((-A)^{1/2})\hookrightarrow L^2(\Omega) =[L^2(\Omega)]'\hookrightarrow
[D((-A)^{1/2} )]'=H^{-1}_{\Gamma_0}(\Omega) ,
\end{equation}
where $H^{-1}_{\Gamma_0}(\Omega)$ is the dual space of $H^{ 1}_{\Gamma_0}(\Omega)$ with the pivot
space  $L^2(\Omega)$.  An extension $\tilde{A}\in
{\mathcal{L}}(H^{1}_{\Gamma_0}(\Omega), H^{-1}_{\Gamma_0}(\Omega))$ of $A$ is defined
by
 \begin{equation}
\label{2020981918}\langle \tilde Ax,z\rangle_{H^{-1}_{\Gamma_0}(\Omega),
 H^{1}_{\Gamma_0}(\Omega)}
 =-\langle (-A)^{1/2}x,(-A)^{1/2}z\rangle_{L^2(\Omega)}, \;\  \forall \; x,z\in H^{1}_{\Gamma_0}(\Omega).
 \end{equation}
  Since $A$ is strictly negative, self-adjoint in $L^2(\Omega)$, and is the inverse of a compact
operator, the operator $A$ has the infinite sequence of negative eigenvalues
$\{\lambda_j\}_{j=1}^{\infty}$
 and a corresponding sequence of eigenfunctions  $\{\phi_j(\cdot)\}_{j=1}^{\infty}$  that forms an orthonormal
basis for $L^2(\Omega)$.
Without loss  of  generality, we always  assume that

 \begin{assumption}\label{Assum202010121559}
 Let the operator $ A$ be given by
   \dref{20209101543} and $\mu>0$.
 Suppose that the eigenpairs   $\{(\phi_j(\cdot), \lambda_j)\}_{j=1}^{\infty}$ of $A$  satisfy
 \begin{equation}\label{202010101700}
0>\lambda_1>\lambda_2>\cdots>\lambda_k>\cdots\to-\infty,
\end{equation}
and
\begin{equation}\label{20209101524}
\left\{\begin{array}{l}
\disp \Delta \phi_k=\lambda_k\phi_k,\ \ \ \|\phi_k\|_{L^2(\Omega)}=1,\crr
\disp  \phi_k(x)=0,\ x\in \Gamma_0,\ \frac{\partial \phi_k(x)}{\partial \nu}=0,\ x\in\Gamma_1,  \ \
\end{array}\right.k=1,2,\cdots.
\end{equation}
Suppose that  $N$  is an integer that satisfies
\begin{equation}\label{201912301959}
   \lambda_k+\mu<0 ,\ \ \forall\ k>N.
 \end{equation}
\end{assumption}

  Define
the Neumann map $\Upsilon\in \mathcal{L}(L^2(\Gamma_1), $ $
H^{3/2}(\Omega))$ (\cite[p. 668]{LT3}) by $\Upsilon u=\psi$ if and
only if
\begin{equation}
\label{4.5} \left\{\begin{array}{l} \Delta \psi=0 \hbox{ in }\Omega,
\crr\disp  \psi|_{\Gamma_0}=0, \;  \frac{\partial
\psi}{\partial\nu}\big|_{\Gamma_1}=u.
\end{array}\right.
\end{equation}
Using  the Neumann map, one can write \dref{2020981903}  in $H^{-1}_{\Gamma_0}(\Omega)$
as
\begin{equation}
\label{2020981933}
\begin{array}{l}
\disp \dot{w}(\cdot,t)  =
 \Delta w(\cdot,t) -\Delta\psi+\mu w(\cdot,t) =\Delta( w(\cdot,t)-\psi)+\mu w(\cdot,t)\crr
 \disp =
 \tilde{A}(w(\cdot,t)-\psi)+\mu w(\cdot,t)=
 (\tilde{A}+\mu )  w(\cdot,t) -\tilde{A} \Upsilon u(\cdot,t).
 \end{array}
\end{equation}
That is
\begin{equation}
\label{2020981943}  \dot{w}(\cdot,t)= (\tilde{A}+\mu )  w(\cdot,t)+Bu(\cdot,t) \quad\mbox{in}\quad H_{\Gamma_0}^{-1}(\Omega),
\end{equation}
 where $B\in {\mathcal{L}}(L^2(\Gamma_1),H_{\Gamma_0}^{-1}(\Omega))$ is given by
\begin{equation}\label{4.8}
 Bu=-\tilde{A}\Upsilon u, \; \forall \; u\in L^2(\Gamma_1).
\end{equation}
Define  $B^*\in {\cal L}(H^1_{\Gamma_0}(\Omega),L^2(\Gamma_1))$   by
 \begin{equation}
\label{2020910949ad1012}
 \langle B^{*}f,u\rangle_{L^2(\Gamma_1)}=\langle
f,B u\rangle_{H^{1}_{\Gamma_0}(\Omega),  H^{-1}_{\Gamma_0}(\Omega) }, \forall \; f\in
H^{1}_{\Gamma_0}(\Omega), u\in L^2(\Gamma_1).
\end{equation}
Then, for any $f\in D(A)=D(A^*)$ and $u\in L^2(\Gamma_1)$,  it follows from \dref{2020981918}, \dref{20209101543},   \dref{2020910949ad1012} and \dref{4.5}  that
 \begin{equation}\label{20201012909}
 \begin{array}{l}
\disp  \langle B^{*}f,u\rangle_{L^2(\Gamma_1)}=\langle
f,-\tilde{A}\Upsilon u\rangle_{H^{1}_{\Gamma_0}(\Omega),  H^{-1}_{\Gamma_0}(\Omega) } =
\langle
A^*f,- \Upsilon u\rangle_{L^2(\Omega) } \crr
 \disp =
\langle
\Delta f,- \psi\rangle_{L^2(\Omega) } =\langle
\nabla  f,\nabla \psi\rangle_{L^2(\Omega) }=\int_{\Gamma_1}f(x)u(x)dx,
\end{array}
\end{equation}
which, together with the denseness of $D(A)$ in $ H^{1}_{\Gamma_0}(\Omega)$,    implies that
 \begin{equation}
\label{2020910949}
  B^{*}f =f|_{\Gamma_1},  \ \ \ \forall \; f\in
H^{1}_{\Gamma_0}(\Omega).
\end{equation}
Using the operators $A$, $B$ and $B^*$, the control plant \dref{2020981903} can be written abstractly
\begin{equation}\label{2021162137}
 \left\{\begin{array}{l}
\disp   \dot{w}(\cdot,t)= (\tilde{A}+\mu ) w(\cdot,t)+B u(\cdot,t) , \ \ t>0,\crr
 y(\cdot,t)=B^* w(\cdot,t),\ \ t\geq0.
 \end{array}\right.
\end{equation}

The rest of the paper is organized as follows:
In Section \ref{Se.2}, we give a   spectral truncation stabilizer
that will be used in   the full state feedback design   in Section  \ref{Se.3}.
The exponential stability of the closed-loop system  is also proved in Section  \ref{Se.3}.
 Section \ref{Se.4}  gives  some  preliminary results
  about the  observer design.  Section \ref{Se.5} is devoted to the observer design and its well-posedness proof.
 Section \ref{Numerical} presents some numerical simulations,
 followed up conclusions in Section \ref{Conclusions}. For the sake of readability,
  some results that are less relevant to the feedback or observer
design are arranged  in the Appendix.

Throughout the paper,
the identity matrix on the   space $\R^n$  will be denoted by $I_n$.
The space of bounded linear operators from $X_1$ to $X_2$ is denoted by $\mathcal{L}(X_1, X_2 )$.
The space of bounded linear operators from $X$ to itself is denoted by $\mathcal{L}(X)$.
The spectrum, resolvent set    and the domain  of
the  operator $A$ are  denoted by $\sigma(A)$, $\rho(A)$,
  and $D(A)$, respectively.

\section{ A spectral truncation stabilizer }\label{Se.2}
 This section is devoted to the  preliminaries on  the state feedback design.
 Suppose that $p\in L^2(\Gamma_1)$ such that
\begin{equation}\label{20201091629}
 \int_{\Gamma_1}p(x)\phi_j(x)  dx  \neq 0,\ \ j=1,2,\cdots, N,
\end{equation}
where $\phi_j$ is  given by \dref{20209101524} and $N$ is an integer that satisfies \dref{201912301959}. The existence of such a function $p$ is trivial and is given  by Lemma \ref{Lm202010101612} in Appendix. In terms of the function $p$, we can define  the operator $P_p :\R\to L^2(\Omega)$ by
\begin{equation}\label{202010121608}
P_p  \theta=  \zeta_p ,\ \ \forall\ \ \theta\in \mathbb{R},
\end{equation}
where
$\zeta_p $ is the solution of the following system:
\begin{equation}\label{202010121609}
 \left\{\begin{array}{l}
\disp \Delta \zeta_p  =\theta\zeta_p \ \ \mbox{in}\ \ \Omega,  \crr
\disp  \zeta_p (x)=0,\ x\in\Gamma_0,\ \ \frac{\partial \zeta_p (x)}{\partial \nu}=  p(x)  ,\ x\in\Gamma_1.
 \end{array}\right.
\end{equation}

 \begin{lemma}\label{lm202010131038}
In addition to Assumption \ref{Assum202010121559},
 suppose that $p\in L^2(\Gamma_1)$ satisfies \dref{20201091629} and suppose that $\theta\in \R$ satisfies \begin{equation}\label{2020101311039}
 \theta \neq \lambda_j ,\ \ j=1,2,\cdots, N.
\end{equation}
   Then,
 the operator $P_{p}$ defined by \dref{202010121608} satisfies
  \begin{equation}\label{2020101310407}
 \langle P_p{\theta }, \phi_j\rangle _{L^2(\Omega)} \neq 0,\ \ j=1,2,\cdots, N.
\end{equation}

\end{lemma}
 \begin{proof}

 It follows from \dref{20209101524},     \dref{202010121608}  and \dref{202010121609} that
 \begin{equation}\label{202010131050}
 \begin{array}{rl}
\disp  \theta\int_{ \Omega } \zeta_p(x)\phi_j (x)dx&\disp =
 \int_{ \Omega } \Delta \zeta_p(x)\phi_j (x)dx\crr
 &\disp =
 \int_{ \Gamma_1 }  p(x)\phi_j (x)dx-\int_{ \Omega } \nabla \zeta_p(x) \nabla\phi_j (x)dx\crr
 &\disp =
 \int_{ \Gamma_1 }  p(x)\phi_j (x)dx+\lambda_j \int_{ \Omega }   \zeta_p(x)  \phi_j (x)dx,
 \end{array}
\end{equation}
which yields
 \begin{equation}\label{202010131051}
\disp  \int_{ \Omega } \zeta_p(x)\phi_j (x)dx= \frac{1}{  \theta -\lambda_j }
 \int_{ \Gamma_1 }  p(x)\phi_j (x)dx \neq0,\ \ j=1,2,\cdots,N.
\end{equation}
The proof is complete due to \dref{20201091629}.
\end{proof}

For any $\theta\in\R$, we   consider the stabilization of system $(A+\mu, P_p \theta )$
that is associated with the following
system:
 \begin{equation}\label{20209101452}
\left\{\begin{array}{l}
 z_{t}(x,t)= \Delta z (x,t)+\mu z(x,t)+(P_p \theta)(x) u(t),\ \   x\in \Omega, \ t> 0,\crr
 z (x,t)=0,\ \ \ x\in\Gamma_0,\ \ \disp \frac{\partial z(x,t)}{\partial \nu}=0,\ \ x\in\Gamma_1,
\end{array}\right.
\end{equation}
where $p\in L^2(\Gamma_1)$ satisfies \dref{20201091629} and $u$ is a scalar control.
Since  $\{\phi_j(\cdot)\}_{j=1}^{\infty}$   defined
by Assumption \ref{Assum202010121559}
 forms an orthonormal
basis for $L^2(\Omega)$,
the function $P_p \theta $ and the solution  $z(\cdot,t)$ of \dref{20209101452}
can be represented  respectively  as
\begin{equation}\label{20209101520}\left.\begin{array}{l}
\disp P_p \theta =\sum\limits_{k=1}^{\infty}f_k \phi_k ,\ \ f_k =\displaystyle \int_{ \Omega } (P_p \theta )(x)\phi_k(x)dx, \ \ k=1,2,\cdots
\end{array}\right.\end{equation}
and
 \begin{equation}\label{2020881609}\left.\begin{array}{l}
\disp  z (\cdot,t)=\sum\limits_{k=1}^{\infty}z_k(t)\phi_k(\cdot),  \ \ z_k(t)=\displaystyle \int_{ \Omega } z (x,t)\phi_k(x)dx,\ \  k=1,2,\cdots.
\end{array}\right.\end{equation}



Inspired by   \cite{CoronTrelat2004SICON,PrieurandTrelat2019TAC} and similarly to \cite{FPart1},
system \dref{20209101452} can be   stabilized by the finite-dimensional spectral
truncation technique.
Actually, by   a simple computation, it follows   that
\begin{equation}\label{wxh201912305} \begin{array}{rl}
\dot{z}_k(t)=&\displaystyle \int_{ \Omega } z _t(x,t)\phi_k(x)dx
=\displaystyle \int_{ \Omega }\left[\Delta z  (x,t)+  \mu z (x,t)+(P_p \theta)(x)u(t)\right]\phi_k(x)dx\crr
=&( \lambda_k+  \mu )z_k(t)+f_ku(t).
\end{array} \end{equation}
Since   $z_k(t)$ is stable for all $k>N$,  where $N$ is given by \dref{201912301959},  it is therefore sufficient   to consider
$z_k(t)$  for  $k\leq N$, which satisfy the following finite-dimensional system:
\begin{equation}\label{201912302011}
\dot{Z}_N(t)=\Lambda_NZ_N(t)+F_Nu(t),\ \ Z_N(t)=(z_1(t),\cdots, z_{N}(t))^\top,
\end{equation}
where $\Lambda_N$  and $F_N$ are defined by
\begin{equation}\label{2020871455}\left\{\begin{array}{l}
\Lambda_N={\rm diag} ( \lambda_1+ \mu,\cdots, \lambda_{N}+ \mu) ,\crr
F_N=\left(f_1,f_2,\cdots,f_{N} \right)^\top .
\end{array}\right.\end{equation}
In this way, the stabilization of system \dref{20209101452} amounts to stabilizing
the finite-dimensional system
\dref{201912302011}.

\begin{lemma}\label{Th201912302046}
In addition to  Assumption \ref{Assum202010121559}, suppose that
$p\in L^2(\Gamma_1)$ satisfies \dref{20201091629} and suppose that $\theta\in\R$ satisfies \dref{2020101311039}.
Then, there exists an $L_N=(l_1,l_2,\cdots,l_N)\in \mathcal{L}(\R^N,\R)$ such that  $\Lambda_N{+}F_NL_N$
 is Hurwitz, where $\Lambda_N$  and $F_N$ are defined by
\dref{2020871455}.
Moreover,
 the operator $ A+{\mu}+(P_{p} \theta)  K  $ generates an exponentially stable $C_0$-semigroup
 on $L^2(\Omega)$,
where $P_p\theta$ is given by \dref{202010121608} and
$ K  $ is  given by
 \begin{equation}\label{2020881654}
 K : g\to   \int_{\Omega} g(x) \left[\sum_{k=1}^{N}l_k\phi_k(x)\right] dx,\ \ \forall\ g\in L^2(\Omega).
\end{equation}

\end{lemma}
\begin{proof}
 Owing to
  \dref{20201091629}, it follows from
  Lemma \ref{lm202010131038}  that \dref{2020101310407} holds. By Lemma \ref{Lm202010131720} in Appendix,
  the pair $(\Lambda_N,F_N)$ is controllable and hence,
  there exists a vector $L_N=(l_1,l_2,\cdots,l_N) $ such that
  $\Lambda_N+F_NL_N$ is Hurwitz.

Since $A+{\mu}$ generates   an analytic semigroup $e^{(A+{\mu})t}$  on $L^2(\Omega)$ and
$ (P_p\theta)  K  \in\mathcal{L}(L^2(\Omega))$, it follows from \cite[Corollary 2.3, p.81]{Pazy1983Book} that
 $  A+{\mu}+( P_p\theta)  K  $  also  generates an analytic semigroup on $L^2(\Omega)$.
As a result, the proof will be accomplished if we can show that  $\sigma( A+{\mu}+ (P_p\theta)  K  )\subset \{ s\ | \ {\rm Re}(s)<0\}$. For any $\lambda\in \sigma( A+{\mu}+ (P_p\theta)  K  )$, we consider the characteristic equation
$ (A+{\mu}+ (P_p\theta)  K ) g=\lambda g $ with $g\neq 0 $.

When $g\in  {\rm Span}\{\phi_1,\phi_2,\cdots,\phi_N\} $,  there exist $g_1,g_2,\cdots,g_N\in \R$ such that  $g=\sum_{j=1}^{N}g_j\phi_j$. The   characteristic equation
becomes
 \begin{equation}\label{201912312136}
 \sum_{j=1}^Ng_j (A+{\mu}) \phi_j + P_p{\theta}\sum_{j=1}^N g_j K \phi_j= \sum_{j=1}^N\lambda g_j\phi_j.
\end{equation}
Since $(A+{\mu})\phi_j=( \lambda_j+  \mu)\phi_j$ and
 \begin{equation}\label{201912312139}
   K \phi_j=    \int_{0}^{1} \phi_j(x)   \left[\sum_{k=1}^{N}l_k\phi_k(x)\right]dx
  =l_j  ,\ \ j=1,2,\cdots,N,
\end{equation}
the equation \dref{201912312136} takes the  form
\begin{equation}\label{201912312141}
 \sum_{j=1}^Ng_j( \lambda_j+ \mu)\phi_j + P_p\theta\sum_{j=1}^N g_jl_j= \sum_{j=1}^N\lambda g_j\phi_j.
\end{equation}
Take the  inner product with $\phi_k $, $k=1,2,\cdots,N$ on equation \dref{201912312141}
 to obtain
\begin{equation}\label{wxh201912312146}
 g_k( \lambda_k+ \mu)  +  f_k\sum_{j=1}^N g_jl_j=  \lambda g_k,\ \ k=1,2,\cdots,N,
\end{equation}
which, together with \dref{2020871455}, leads to
\begin{equation}\label{201912312148}
 (\lambda -\Lambda_N-F_NL_N) \begin{pmatrix}
                       g_1\\g_2\\\vdots\\g_N
                     \end{pmatrix} =0.
\end{equation}
Since $(g_1,g_2,\cdots,g_N)\neq 0$, we have
\begin{equation}\label{wxh201912312149}
{\rm Det}(\lambda -\Lambda_N-F_NL_N) =0.
\end{equation}
 Hence, $\lambda\in \sigma(\Lambda_N{+}F_NL_N) \subset\{ s\ | \ {\rm Re} (s) <0\}$, since $\Lambda_N{+}F_NL_N$
 is Hurwitz.

  When   $g\notin {\rm Span}\{\phi_1,\phi_2,\cdots,\phi_N\}$,   there exists a $j_0>N$ such  that $\displaystyle \int_{0}^{1}g(x)\phi_{j_0}(x)dx\neq 0$.
Take the  inner product with $\phi_{j_0}$ on equation $ (A+{\mu}+ P_p\theta  K ) g=\lambda g$
 to get \begin{equation}\label{wxh2020322257}
( \lambda_{j_0}+\mu) \int_{0}^{1}g(x)\phi_{j_0}(x)dx=  \lambda \int_{0}^{1}g(x)\phi_{j_0}(x)dx,
\end{equation}
which implies that   $ \lambda = \lambda_{j_0}+\mu<0$. Therefore,  $\lambda\in \sigma( A+{\mu}+ (P_p\theta)  K  )\subset \{ s\ | \ {\rm Re}(s)<0\}$.
 The proof is  complete.
 \end{proof}

\section{State feedback}\label{Se.3}
This section is devoted to the stabilization of    system \dref{2020981903}.
Inspired by \cite{FPart1}, we consider the following dynamics feedback:
 \begin{equation}\label{2020981950}
 \left\{\begin{array}{l}
u(x,t)= v(x,  t),\ \  x\in \Gamma_1,\crr
\disp
{v}_t (\cdot, t)=-\alpha v(\cdot,t)+ B_v  u_v(t) \ \ \mbox{in}\ \ L^2(\Gamma_1),
\end{array}\right.t\geq0,
\end{equation}
where $\alpha >0$  is a tuning parameter,  $u_v(t)\in \R$  is a     new  scalar    control to be designed  and
 the operator $B_v\in\mathcal{L}(\R,L^2(\Gamma_1)) $    is given by
 \begin{equation}\label{20201091632}
B_v c= c p(\cdot),\ \ \forall\ \ c\in \R,
\end{equation}
with $p\in L^2(\Gamma_1)$ satisfying \dref{20201091629}.
  Under the controller \dref{2020981950}, the control plant
      \dref{2021162137}, or equivalently  \dref{2020981903},   turns to be
\begin{equation}\label{2020910729}
 \left\{\begin{array}{l}
\disp   \dot{w}(\cdot,t)= (\tilde{A}+\mu ) w(\cdot,t)+B v(\cdot,t) \quad\mbox{in}\quad H_{\Gamma_0}^{-1}(\Omega), \crr
\disp  {v}_t (\cdot, t)=-\alpha v(\cdot,t)+ B_v  u_v(t)   \quad\mbox{in}\quad L^2(\Gamma_1).
 \end{array}\right.
\end{equation}
Since \dref{2020910729} is a cascade system, the ``$v$-part" can be regarded as the   actuator dynamics
of the control plant $w$-system.
As a result, we can stabilize system \dref{2020910729} by the newly developed actuator dynamics compensation  approach in  \cite{FPart1}.
To  demonstrate the key  idea of controller  design clearly, we first consider the following  finite-dimensional example.
\begin{example}\label{Ex202010131134}
 Consider the following system in the state space $\R^n\times \R$:
\begin{equation}\label{202010131135}
\left\{\begin{array}{l}
 \dot{x}_1(t)=A  x_1( t)+B  x_2(t),\crr
   \dot{x}_2(t)=-\alpha x_2( t)+ B_2u(t),
     \end{array}\right.\ \ \alpha>0,
\end{equation}
where $A \in  \R^{n\times n} $, $B  \in   \R^{n}   $,
    $B_2\in  \R  $  and   $u(t)$ is the control.
By \cite{FPart1},
if we choose $S$ specially such that
   \begin{equation} \label{20191024842Ad1013}
A S+\alpha S =B  ,
\end{equation}
then
 system \dref{202010131135} can be  decoupled   by  the
block-upper-triangular transformation:
\begin{equation} \label{20191024837Ad1013}
\begin{pmatrix}
I_n&S\\
0&1
\end{pmatrix}
\begin{pmatrix}
A &B  \\
0&-\alpha
\end{pmatrix}\begin{pmatrix}
I_n& S\\
0&1
\end{pmatrix}^{-1}
=
\begin{pmatrix}
A &0  \\
0&-\alpha
\end{pmatrix}.
\end{equation}
Hence,
  the controllability of the
following pairs is  equivalent:
\begin{equation} \label{201910171043Ad1013}
\left(\begin{pmatrix}
A & B   \\
0&-\alpha
\end{pmatrix},\
 \begin{pmatrix}
0\\
B_2
\end{pmatrix}\right)\ {\rm  and }\
\left(\begin{pmatrix}
A&0\\
0&-\alpha
\end{pmatrix},\
\begin{pmatrix}
SB_2 \\
B_2
\end{pmatrix}
\right).
\end{equation}
Owing to the block-diagonal structure, the stabilization of the second  system of \dref{201910171043Ad1013} is
much   easier than the first one. As a consequence of this fact,
  the controller  $u (t)$ in \dref{202010131135} can be designed  by   stabilizing  system $(A ,SB_2)$:
  \begin{equation} \label{201910241002Ad1013}
u (t)=(K ,0)\begin{pmatrix}
I_n&S\\
0&1
\end{pmatrix}
\begin{pmatrix}
x_1(t)\\
x_2(t)
\end{pmatrix}
=K Sx_2(t)+K x_1(t),
\end{equation}
  where
$K \in \R^{1\times n}$   is chosen to make  $A +SB_2K $ Hurwitz.
 Under the feedback \dref{201910241002Ad1013}, we obtain
  the closed-loop
 of system \dref{202010131135}:
 \begin{equation} \label{202010131456}
 \left\{\begin{array}{l}
\disp \dot{ x}_1  (t) = A  x_1(t)+B  x_2(t),\crr
\disp  \dot{x}_2(t) =   ( B_2 K S-\alpha)x_2(t) +B_2 K {x}_1(t),
\end{array}\right.
\end{equation}
which is stable due to the Hurwitz matrix $A +SB _2K$ and the similarity
\begin{equation} \label{201910201037A1}
 \begin{pmatrix}
A & B   \\
 B_2K &B_2 K S-\alpha
\end{pmatrix}
 \sim
 \begin{pmatrix}
A +SB _2K & 0 \\
 B_2 K &  -\alpha
\end{pmatrix}.
\end{equation}
 To sum up, the feedback of system \dref{202010131135}   can be designed by  the following scheme: (i), solve the equation \dref{20191024842Ad1013}  to
get $S$; (ii), choose $K$  such that  $A +SB _2K$   is Hurwitz; (iii), let
$u(t)=K Sx_2(t)+K x_1(t)$.
\end{example}

Now, we return to the feedback design of system \dref{2020910729}. Inspired by Example \ref{Ex202010131134}, the controller can be designed as
 \begin{equation}\label{202010131506}
 u_v(t)=K  {w}(\cdot,t)+KSv(t),
 \end{equation}
where
$S\in \mathcal{L}( L^2(\Gamma_1), L^2(\Omega))$ solves the
 Sylvester equation  \begin{equation}\label{2020982157}
(\tilde{A}+\mu)S+\alpha S =B,
\end{equation}
and $K\in \mathcal{L}(L^2(\Omega),\R)$ stabilizes system $(A+\mu, SB_v)$ exponentially in the sense of \cite{Weiss1997TAC}.
\begin{lemma}\label{lm202010131514}
 Let $ A$ and $B$ be given by
   \dref{20209101543} and \dref{4.8}, respectively. Suppose that
$B_v\in\mathcal{L}(\R,L^2(\Gamma_1)) $   is given by \dref{20201091632}
with $p\in L^2(\Gamma_1)$ satisfying \dref{20201091629} and suppose that
 \begin{equation}\label{2020911847}
 \alpha+\mu\in \rho(-A).
 \end{equation}
 Then, the solution of Sylvester equation  \dref{2020982157} satisfies
   \begin{equation}\label{202010131545}
S g= - \varphi_g\in L^2(\Omega),\ \ \forall\ g\in L^2(\Gamma_1),
\end{equation}
 where  $\varphi_g$ is given by
 \begin{equation}\label{202010131549}
\left\{\begin{array}{l}
  \disp \Delta\varphi_g  =(-\alpha-\mu )\varphi_g  \ \  \mbox{in}\ \ \Omega,   \crr
\disp  \varphi_g (x )=0,\ x\in\Gamma_0,\ \ \frac{\partial \varphi_g (x )}{\partial \nu}=  g(x)  ,\ x\in\Gamma_1.
 \end{array}\right.
\end{equation}
 Moreover, for any $c\in\R$, we have
 \begin{equation}\label{202091091635}
S B_v c =-c P_p\theta \ \ \mbox{with}\ \ \theta=-\alpha-\mu,
\end{equation}
 where $P_p :\R\to L^2(\Omega)$ is given by \dref{202010121608}.
 \end{lemma}
\begin{proof}
Owing to \dref{2020911847}, we solve   \dref{2020982157} to get
  \begin{equation}\label{2020910832}
S = (\alpha +\mu +\tilde{A})^{-1} B.
\end{equation}
By a straightforward computation, it follows that
\begin{equation}\label{202010131621}
 \left.\begin{array}{ll}
\disp (\alpha +\mu +\tilde{A}) \varphi_g   &\disp=(\alpha +\mu +\tilde{A})  \varphi_g  - \tilde{A}\Upsilon   g+ \tilde{A}\Upsilon   g\crr
&\disp =
(\alpha +\mu )\varphi_g  +\tilde{A}(\varphi_g  -  \Upsilon   g)+ \tilde{A}\Upsilon  g
\crr
&\disp= (\alpha +\mu )\varphi_g  +\Delta(\varphi_g  -  \Upsilon   g )+ \tilde{A}\Upsilon   g\crr
&\disp =  \tilde{A}  \Upsilon   g=-B  g ,
 \end{array}\right.
\end{equation}
which, together with \dref{2020910832}, leads to \dref{202010131545} easily.

 By \dref{20201091632} and \dref{202010131545}, we have $SB_vc= -c \vartheta$, where
\begin{equation}\label{202010131632}
\left\{\begin{array}{l}
  \disp \Delta \vartheta   =(-\alpha-\mu ) \vartheta   \ \  \mbox{in}\ \ \Omega,    \crr
\disp  \vartheta  (x )=0,\ x\in\Gamma_0,\ \ \frac{\partial \vartheta (x )}{\partial \nu}=   p(x)  ,\ x\in\Gamma_1.
 \end{array}\right.
\end{equation}
In view of   \dref{202010121609} and letting $\theta=-\alpha-\mu$,
we can obtain  \dref{202091091635} easily.
The proof is complete.
\end{proof}

By Lemmas \ref{Th201912302046} and \ref{lm202010131514},
the operator  $-K\in \mathcal{L}(L^2(\Omega),\R)$  defined by  \dref{2020881654}
 stabilizes system $(A+\mu, SB_v)$ exponentially. As a result, the controller \dref{202010131506} turns to be
  \begin{equation}\label{202010131538}
 u_v(t)=-\int_{\Omega}[ w(x,t)- \varphi_v(x,t)] \left[\sum_{k=1}^{N}l_k\phi_k(x)\right] dx,
\end{equation}
  where
  \begin{equation}\label{202010131542}
\left\{\begin{array}{l}
  \disp \Delta \varphi_v(\cdot,t)  =(-\alpha-\mu )\varphi_v(\cdot,t)  \ \ \mbox{in}\ \ \Omega,  \crr
\disp  \varphi_v (x,t)=0,\ x\in\Gamma_0,\ \ \frac{\partial \varphi_v (x,t)}{\partial \nu}=  v(x,t)  ,\ x\in\Gamma_1.
 \end{array}\right.
\end{equation}
By \dref{202010131538} and \dref{2020910729}, we obtain the closed-loop system
\begin{equation}\label{202010131646}
\left\{\begin{array}{l}
\disp \dot{w}(\cdot,t)= (\tilde{A}+\mu ) w(\cdot,t)+B v(\cdot,t) \quad\mbox{in}\quad \Omega , \crr
\disp  {v}_t (\cdot, t)=-\alpha{v}(\cdot,t)- B_v
 \int_{\Omega} [{w}(x,t) -\varphi_v(x,t)] \left[\sum_{k=1}^{N}l_k\phi_k(x)\right] dx\ \ \mbox{in}\ \ \Gamma_1,\crr
 \disp \Delta \varphi_v(\cdot,t)  =(-\alpha-\mu )\varphi_v(\cdot,t) \ \ \mbox{in}\ \ \Omega, \crr
\disp  \varphi_v (x,t)=0,\ x\in\Gamma_0,\ \ \frac{\partial \varphi_v (x,t)}{\partial \nu}=  v(x,t)  ,\ x\in\Gamma_1.
 \end{array}\right.
\end{equation}
Combining \dref{20201091632},   \dref{4.8}  and \dref{20209101543}, system \dref{202010131646} turns to be
\begin{equation}\label{20201092100}
\left\{\begin{array}{l}
 w_{t}(x,t)= \Delta w (x,t)+\mu w(x,t),\ \   x\in \Omega,  \crr
 w (x,t)=0,\ \ \ x\in\Gamma_0,\ \ \disp \frac{\partial w(x,t)}{\partial \nu}=v(x,t),\ \ x\in\Gamma_1,\crr
\disp  {v}_t (\cdot, t)=-\alpha{v}(\cdot,t)- p(\cdot)
 \int_{\Omega} [{w}(x,t) -\varphi_v(x,t)] \left[\sum_{k=1}^{N}l_k\phi_k(x)\right] dx\ \ \mbox{in}\ \ \Gamma_1,\crr
 \disp \Delta \varphi_v(\cdot,t)  =(-\alpha-\mu )\varphi_v(\cdot,t)  \ \ \mbox{in}\ \ \Omega,\crr
\disp  \varphi_v (x,t)=0,\ x\in\Gamma_0,\ \ \frac{\partial \varphi_v (x,t)}{\partial \nu}=  v(x,t)  ,\ x\in\Gamma_1.
 \end{array}\right.
\end{equation}

 \begin{theorem}\label{Th20201092120}
In addition to Assumption  \ref{Assum202010121559}, suppose that
$p\in L^2(\Gamma_1)$ satisfies
\dref{20201091629} and
\begin{equation}\label{202010131658}
 \alpha +\mu + \lambda_j \neq 0,\ \ j=1,2,\cdots, N.
\end{equation}
Then, there exists an $L_N=(l_1,l_2,\cdots,l_N)\in \mathcal{L}(\R^N,\R)$ such that  $\Lambda_N{+}F_NL_N$
 is Hurwitz, where $\Lambda_N$  and $F_N$ are defined by
\dref{2020871455}. Moreover,
    for any $(w(\cdot,0),v(\cdot,0))^{\top} \in L^2(\Omega)\times L^2(\Gamma_1)$,   system \dref{20201092100} admits a unique solution
 $ (w,v)^{\top}\in C([0,\infty);L^2(\Omega)\times L^2(\Gamma_1)) $
that decays to zero exponentially in $L^2(\Omega)\times L^2(\Gamma_1)$  as $t\to\infty$. Moreover, if the initial state $(w(\cdot,0),v(\cdot,0))^{\top}  \in D(A)\times L^2(\Gamma_1)$, the solution
$ (w,v)^{\top}\in C^1([0,\infty);L^2(\Omega)\times L^2(\Gamma_1)) $ is    classical.
\end{theorem}
\begin{proof}

Notice that  \dref{2020881654}, the closed-loop system  \dref{202010131646}
   can be written as the abstract form:
 \begin{equation}\label{202010121536}
 \frac{d}{dt}( {w} (\cdot,t), {v} (\cdot,t))^{\top}=  \A ( {w} (\cdot,t), {v} (\cdot,t))^{\top} ,
\end{equation}
where the operator $\A: D(\A)\subset L^2(\Omega)\times L^2(\Gamma_1)\to  L^2(\Omega)\times L^2(\Gamma_1)$ is defined by
\begin{equation}\label{202010121537}
 \A=\begin{pmatrix}
A+\mu &B\\
-B_vK&-B_vKS- \alpha
    \end{pmatrix} \ \ \mbox{with}\ \  D(\A)=  D(A)\times L^2(\Gamma_1).
 \end{equation}
 As proposed in \cite{FPart1} and similarly to \dref{20191024837Ad1013}, we
   introduce the following transformation:
  \begin{equation}\label{2020109938}
\mathbb{S}(f,g)^{\top}=(f+Sg,g)^{\top},\ \ (f,g)^{\top}\in L^2(\Omega)\times L^2(\Gamma_1),
\end{equation}
where
$S\in \mathcal{L}( L^2(\Gamma_1), L^2(\Omega))$ solves the
 Sylvester equation  \dref{2020982157}.
 By a simple computation,
  $\mathbb{S}\in \mathcal{L}(L^2(\Omega)\times L^2(\Gamma_1))$ is invertible and its inverse is
 \begin{equation}\label{2020109939}
\mathbb{S}^{-1}(f,g)^{\top}=(f-Sg,g)^{\top},\ \ (f,g)^{\top}\in L^2(\Omega)\times L^2(\Gamma_1) .
\end{equation}
 Moreover, (see, e.g., \cite[Theorem 5.1]{FPart1})
 \begin{equation}\label{202010121544}
 \mathbb{S}\A \mathbb{S}^{-1}=\A_\mathbb{S}, \ \ D(\A_\mathbb{S})= \mathbb{S}D(\A),
\end{equation}
where
\begin{equation}\label{202010101633}
\A_{\mathbb{S}}= \begin{pmatrix}
 {A}+\mu -SB_v   K & 0\\
B_v  K  &-\alpha
    \end{pmatrix},
\end{equation}
$SB_v$ is given by \dref{202091091635}
  and
$ K  $ is  given by \dref{2020881654}.
Since $SB_v=-P_{p} \theta  $ with $\theta=-\alpha-\mu$,
it follows from  Lemma \ref{Th201912302046} that
 the operator $ A+{\mu}+(P_{p} \theta)  K ={A}+\mu -SB_v   K $ generates an exponentially stable $C_0$-semigroup
 on $L^2(\Omega)$. Owing to the block-triangle structure and \cite[Lemma 3.2]{FPart1}, the operator  $\A_{\mathbb{S}}$
 generates an exponentially stable
 $C_0$-semigroup $e^{\A_{\mathbb{S}}t}$ on $L^2(\Omega)\times L^2(\Gamma_1)$.
As a result,
the operator $ \A  $ generates an exponentially stable $C_0$-semigroup
 on $L^2(\Omega)\times L^2(\Gamma_1)$ due to the similarity \dref{202010121544}.
   \end{proof}

\section{Preliminaries on observer design}\label{Se.4}
 This section is devoted to the preliminaries on the observer design.
 Let $q\in L^2(\Gamma_1)$ satisfy
\begin{equation}\label{20201091629q}
 \int_{\Gamma_1}q(x)\phi_j(x)\neq 0,\ \ j=1,2,\cdots, N,
\end{equation}
where $\phi_j$   is given by \dref{20209101524} and $N$ is an integer that satisfies \dref{201912301959}. For any $\gamma\in\R$, define the operator $J_q^{\gamma}: L^2(\Omega)\to\R$ by
 \begin{equation} \label{202010111141}
J_{q}^{\gamma} (  g)=-\int_{\Gamma_1}q(x)\xi_g(x)dx,\ \ \forall\ g\in L^2(\Omega),
\end{equation}
where $\xi_g$ is given by
 \begin{equation} \label{202010111149}
 \left\{\begin{array}{l}
\disp  \Delta  \xi_g=\gamma \xi_g+g\ \ \mbox{in}\ \ \Omega,\crr
\disp  \xi_g(x)=0,\ x\in\Gamma_0,\ \ \frac{\partial \xi_g(x)}{\partial \nu}=0,\ x\in \Gamma _1.
\end{array}\right.
\end{equation}

\begin{lemma}\label{lm202010111155}
Let  $\{(\phi_j, \lambda_j)\}_{j=1}^{\infty}$  be given by \dref{20209101524} and $N$ be an integer that satisfies \dref{201912301959}.
 Suppose that $q\in L^2(\Gamma_1)$ satisfies \dref{20201091629q} and suppose that $\gamma\in \R$ satisfies \begin{equation}\label{202010111209}
 \gamma \neq \lambda_j ,\ \ j=1,2,\cdots, N.
\end{equation}
   Then,
 the operator $J_{q}^{\gamma }$ defined by \dref{202010111141} satisfies
  \begin{equation}\label{202010111157}
J_{q}^{\gamma }(\phi_j)\neq 0,\ \ j=1,2,\cdots, N.
\end{equation}

\end{lemma}
\begin{proof}
  Let  $\eta_q$ be a solution of the following system
   \begin{equation} \label{202010111159}
 \left\{\begin{array}{l}
\disp  \Delta  \eta_q=\gamma\eta_q\ \ \mbox{in}\ \ \Omega,    \crr
\disp  \eta_q(x)=0,\ x\in\Gamma_0,\ \ \frac{\partial \eta_q(x)}{\partial \nu}=q(x),\ x\in \Gamma _1.
\end{array}\right.
\end{equation}
 Then, for any $g\in L^2(\Omega)$,  it follows from \dref{202010111149} and \dref{202010111159} that
 \begin{equation}\label{2020982014}
 \begin{array}{rl}
\disp \gamma \langle \eta_q,\xi_g\rangle_{ L^2(\Omega)}&\disp =\langle \Delta \eta_q,\xi_g\rangle_{ L^2(\Omega)}=\int_{\Gamma_1}\frac{\partial \eta_q(x)}{\partial\nu}\xi_g(x)dx-
\langle \nabla \eta_q,\nabla \xi_g\rangle_{ L^2(\Omega)}\crr
\disp &\disp =
\int_{\Gamma_1}q(x)\xi_g(x)dx+\gamma\langle \eta_q,\xi_g\rangle_{ L^2(\Omega)}+
\langle  g,\eta_q\rangle_{ L^2(\Omega)}
,
 \end{array}
\end{equation}
  which yields
  \begin{equation} \label{202010111240}
J_{q}^{\gamma}(\phi_j)=\langle  \phi_j,\eta_q\rangle_{ L^2(\Omega)},  \ \ j=1,2,\cdots,N.
\end{equation}
   On the other hand,
  \begin{equation}\label{202010111242}
 \begin{array}{rl}
\disp \lambda_j \langle  \phi_j,\eta_q\rangle_{ L^2(\Omega)}&\disp =\langle \Delta\phi_j, \eta_q \rangle_{ L^2(\Omega)}=-\langle \nabla\phi_j, \nabla\eta_q \rangle_{ L^2(\Omega)}\crr
\disp &\disp = - \int_{\Gamma_1}\frac{\partial \eta_q(x)}{\partial \nu}\phi_j(x)dx+
\langle \Delta\eta_q , \phi_j\rangle_{ L^2(\Omega)}\crr
&\disp =-\int_{\Gamma_1}q(x)\phi_j(x)dx+\gamma\langle  \eta_q , \phi_j\rangle_{ L^2(\Omega)}
  \end{array},\ j=1,2,\cdots,N.
\end{equation}
That is
 \begin{equation}\label{202010111251}
 \int_{\Gamma_1}q(x)\phi_j(x)dx=( \gamma-\lambda_j)\langle  \eta_q , \phi_j\rangle_{ L^2(\Omega)}
   ,\ j=1,2,\cdots,N.
\end{equation}
Combining \dref{202010111240}, \dref{202010111251} and \dref{20201091629q}, we obtain \dref{202010111157} easily.
\end{proof}

Next, we will find $K$ to detect  system $(A+\mu, J_{q}^{\gamma})$ exponentially in the sense of \cite{Weiss1997TAC}.
 Define the row vector
  \begin{equation}\label{202010111734}
  J_{N}=(J_{q}^{\gamma}(\phi_1),J_{q}^{\gamma}(\phi_2),\cdots,J_{q}^{\gamma}(\phi_N)),
  \end{equation}
  where  $\phi_i$ is given by \dref{20209101524}, $i=1,2,\cdots,N$ and    $N$ is an integer that satisfies \dref{201912301959}.
 By Lemma \ref{lm202010111155} and Lemma \ref{Lm202010131720} in Appendix,    the finite-dimensional system $(\Lambda_N,J_N)$ is observable, where
  $\Lambda_N$ is given by \dref{2020871455}. As a result, there exists a vector $K_N=(k_1,k_2,\cdots,k_N)^{\top}$ such that
  $\Lambda_N+K_NJ_N$ is Hurwitz.

  \begin{lemma}\label{lm202010111744}
    Suppose that the operator $A $ is  given by \dref{20209101543},
   the eigenpairs   $\{(\phi_j(\cdot), \lambda_j)\}_{j=1}^{\infty}$  satisfy  \dref{20209101524},
       $q(\cdot) \in L^2(\Gamma_1)$ satisfies \dref{20201091629q},
 the  integer $N$  satisfies  \dref{201912301959} and $\mu>0$.
 For any  $\gamma\in \R$ satisfying  \dref{202010111209}, let   $J_q^{\gamma}: L^2(\Omega)\to\R$  be given by
 \dref{202010111141} and $ J_N$ be given by \dref{202010111734}.
Then,
there exists a vector $K_N=(k_1,k_2,\cdots,k_N)^{\top}$ such that
  $\Lambda_N+K_NJ_N$ is Hurwitz,
  where
 $\Lambda_N$  is   given by  \dref{2020871455}.
      Moreover,  the operator $A+\mu+KJ_q^{\gamma}$ generates an exponentially stable $C_0$-semigroup
on $L^2(\Omega)$, where
    the operator  $K: \R\to L^2(\Omega)$  is given  by
     \begin{equation}\label{202010111744}
Kc=c\sum_{j=1}^{N}k_j\phi_j(\cdot),\  \ \forall\ c\in\R.
\end{equation}

  \end{lemma}
  \begin{proof}
  Owing to
  \dref{20201091629q}, it follows from
  Lemma \ref{lm202010111155}  that \dref{202010111157} holds. By Lemma \ref{Lm202010131720} in Appendix,
   the pair $(\Lambda_N,J_N)$ is observable and
  there exists a vector $K_N=(k_1,k_2,\cdots,k_N)^{\top}$ such that
  $\Lambda_N+K_NJ_N$ is Hurwitz.

Since $A+\mu$ generates   an analytic semigroup $e^{(A+\mu)t}$  on $L^2(\Omega)$ and
$ KJ_q^{\gamma}$ is bounded, it follows from \cite[Corollary 2.3, p.81]{Pazy1983Book} that
 $  A+\mu+ KJ_q^{\gamma} $  also  generates an analytic semigroup on $L^2(\Omega)$.
The proof will be accomplished if we can show that  $\sigma(A+\mu+ KJ_q^{\gamma})\subset \{ s\ | \ {\rm Re}(s)<0\}$. For any $\lambda\in \sigma( A+\mu+ KJ_q^{\gamma})$, we consider the characteristic equation
$ (A+\mu+ KJ_q^{\gamma}) g=\lambda g $ with $g\neq 0 $.

When $g\in  {\rm Span}\{\phi_1,\phi_2,\cdots,\phi_N\} $,  set $g=\sum_{j=1}^Ng_j\phi_j$. The   characteristic equation
becomes
 \begin{equation}\label{202010111808}
 \sum_{j=1}^N(  \lambda_j  +{\mu})g_j\phi_j +  \left[\sum_{j=1}^{N}g_j J_{q}^{\gamma} (\phi_j)\right] \sum_{j=1}^N k_j  \phi_j= \sum_{j=1}^N\lambda g_j\phi_j.
\end{equation}
Take the  inner product with $\phi_i $, $i=1,2,\cdots,N$ on equation \dref{202010111808}
 to obtain
\begin{equation}\label{202010111811}
(  \lambda_i  + \mu)g_i  +  k_i\sum_{j=1}^{N}g_j J_{q}^{\gamma} (\phi_j) =  \lambda g_i,\ \ i=1,2,\cdots,N,
\end{equation}
which, together with \dref{2020871455} and \dref{202010111734}, leads to
\begin{equation}\label{202010111812}
 (\lambda -\Lambda_N-K_NJ_N) \begin{pmatrix}
                       g_1\\g_2\\\vdots\\g_N
                     \end{pmatrix} =0.
\end{equation}
Since $(g_1,g_2,\cdots,g_N)\neq 0$, we have
\begin{equation}\label{202010111813}
{\rm Det}(\lambda -\Lambda_N-K_NJ_N) =0.
\end{equation}
 Hence, $\lambda\in \sigma(\Lambda_N{+}K_NJ_N) \subset\{ s\ | \ {\rm Re} (s) <0\}$, since $\Lambda_N{+}K_NJ_N$
 is Hurwitz.

  When   $g\notin {\rm Span}\{\phi_1,\phi_2,\cdots,\phi_N\}$,   there exists a $j_0>N$ such  that $\displaystyle \int_{0}^{1}g(x)\phi_{j_0}(x)dx\neq 0$.
Take the  inner product with $\phi_{j_0}$ on equation $ (A+{\mu}+ K J_q^{\gamma}) g=\lambda g$
 to get \begin{equation}\label{202010111814}
( \lambda_{j_0}+\mu) \int_{0}^{1}g(x)\phi_{j_0}(x)dx=  \lambda \int_{0}^{1}g(x)\phi_{j_0}(x)dx,
\end{equation}
which, together with \dref{201912301959},   implies that   $ \lambda = \lambda_{j_0}+\mu<0$. Therefore,
 $\lambda\in \sigma( A+{\mu}+ K J_q^{\gamma} )\subset \{ s\ | \ {\rm Re}(s)<0\}$.
  The proof is  complete.
 \end{proof}

%
%
%

%
%
%
%

\section{Observer design}\label{Se.5}
This section is devoted to the observer design by the newly proposed approach in \cite{FPart2}.
Instead of the system \dref{2020981903}, we   design the observer for the following system:
\begin{equation}\label{202010101654}
\left\{\begin{array}{l}
 w_{t}(x,t)= \Delta w (x,t)+\mu w(x,t),\ \   x\in \Omega, \crr
 w (x,t)=0,\ \ \ x\in\Gamma_0,\ \ \disp \frac{\partial w(x,t)}{\partial \nu}=u(x,t),\ \ x\in\Gamma_1,\crr
 v_t(x,t)=-\beta v(x,t)+Q   B^*    w(x,t),\ \ x\in\Gamma_1,\crr
 \disp y_v(t)=\int_{\Gamma_1}v(x,t)dx,
  \end{array}\right.
\end{equation}
where $\beta >0$ is a tuning parameter, $v(\cdot,t)$ is an extended state,
$y_v$ is a new output, $B^*$ is given by \dref{2020910949}  and
$Q \in \mathcal{L}(L^2(\Gamma_1))$      is given by
 \begin{equation}\label{20201011931}
(Q  g)(x)= q(x)g(x),\ \ x\in \Gamma_1,\ \ \forall\ \ g\in L^2(\Gamma_1)
\end{equation}
with  $q\in L^2(\Gamma_1)$ satisfying
 \dref{20201091629q}.
 By \dref{20209101543}  and \dref{4.8}, system \dref{202010101654} can be written as
\begin{equation}\label{202010101753}
\left\{\begin{array}{l}
 w_{t}(\cdot,t)=(\tilde{A}+\mu)w (\cdot,t)+B u(\cdot,t),\crr
   v_t(\cdot,t)=-\beta v(\cdot,t)+QB^*w(\cdot,t),\crr
   y_v(t)=C_vv (\cdot,t),
  \end{array}\right.
\end{equation}
where $C_v:L^2(\Gamma_1)\to\R$ is defined by
\begin{equation}\label{202010111134}
 C_vh=\int_{\Gamma_1}h(x)dx,\ \ \forall\ h\in L^2(\Gamma_1).
\end{equation}

Now we  demonstrate the key  idea of the  observer design via a simple finite-dimensional example.
\begin{example}\label{Ex202010111134}
 Consider the following system in the state space $\R^n\times \R$:
\begin{equation}\label{20201012750}
\left\{\begin{array}{l}
 \dot{x}_1(t)=A  x_1( t)+B  u(t),\crr
   \dot{x}_2(t)=-\beta x_2( t)+ QB^*x_1(t),\crr
   y (t)= x_2 ( t),
  \end{array}\right.\ \ \beta>0,
\end{equation}
where $A \in  \R^{n\times n} $  is the system matrix, $B  \in   \R^{n}   $ is the control
 matrix,
    $Q\in  \R  $ is a constant, $u(t)$ is the control and
  $y(t)$ is the measurement.
 The Luenberger observer  of system  \dref{20201012750}  is designed as
 \begin{equation} \label{Fh20202061740}
 \left\{\begin{array}{l}
\disp \dot{\hat{x}}_1(t) =  A \hat{x}_1(t)+F_1[x_2(t)-  \hat{x}_2(t)]+Bu(t),\crr
\disp  \dot{\hat{x}}_2(t) =-\beta\hat{x}_2(t)+ QB^*\hat{x}_1(t)-F_2[x_2(t)- \hat{x}_2(t)],
\end{array}\right.
\end{equation}
where $F_1\in \R^n$ and $F_2\in \R$ are  the gain parameters to be determined.
To demonstrate the key  idea of the  observer design for the infinite-dimensional system \dref{202010101654}, we will find a new way to choose $F_1$ and $F_2$ rather than the
conventional pole placement theorem.
 Let
\begin{equation} \label{Fh20202061802}
\tilde{x}_j(t)=x_j(t)-\hat{x}_j(t),\ \ j=1,2,
\end{equation}
then the error is governed by
\begin{equation} \label{Fh20202061803}
 \left\{\begin{array}{l}
\disp \dot{\tilde{x}}_1(t) = A\tilde{x}_1(t)-F_1 \tilde{x}_2(t),\crr
\disp  \dot{\tilde{x}}_2(t) =  -\beta\tilde{x}_2(t)+QB^*\tilde{x}_1(t)+F_2 \tilde{x}_2(t).
\end{array}\right.
\end{equation}
If we   pick $F_1$ and $F_2$ properly such that  system \dref{Fh20202061803} is stable, then
   $(x_1,x_2)$  can be estimated in the sense that
\begin{equation} \label{Fh20202062106}
  \|(\hat{x}_1(t)-x_1(t) ,\hat{x}_2(t)-x_2(t))\|_{\R^n\times \R}\to  0
  \ \ \mbox{as}\ \ t\to\infty.
\end{equation}
 Inspired by   \cite{FPart2}, the $F_1$ and $F_2$ can be chosen
 easily by decoupling  the   system \dref{Fh20202061803} as a cascade system.
 Consider the following  transformation
\begin{equation} \label{Fh20202061804}
\begin{array}{l}
\begin{pmatrix}
I_n&0\\
P&1
\end{pmatrix}
\begin{pmatrix}
A &-F_1 \\
 QB^*&F_2-\beta
\end{pmatrix}\begin{pmatrix}
I_n&0\\
P&1
\end{pmatrix}^{-1}\crr
=
\begin{pmatrix}
A + F_1P& -F_1 \\
P(A+F_1P)+QB^*-(F_2-\beta)P&F_2-\beta -PF_1
\end{pmatrix},
\end{array}
\end{equation}
where $P\in \R^{1\times n}$ to be determined.
If we choose
\begin{equation} \label{20201012815}
F_2=PF_1\ \ \mbox{and}\ \ P A +QB^*+\beta P=0,
\end{equation}
then the matrix on the right side of \dref{Fh20202061804} is reduced to
\begin{equation} \label{20201012819}
 \begin{pmatrix}
A + F_1P& -F_1 \\
 0& -\beta
\end{pmatrix},
\end{equation}
which is obviously a Hurwitz matrix provided   $A + F_1P$ is Hurwitz.
To   sum up, the tuning parameters $F_1$ and $F_2$ can be chosen by the following scheme:
(i), solve the equation  $P A +QB^*+\beta P=0$ to get $P$; (ii), choose
$F_1$ such that  $A + F_1P$ is Hurwitz; (iii), let $F_2=PF_1$.
 \end{example}

 Now, we return to  design an observer for system \dref{202010101753}.
 Inspired by Example \ref{Ex202010111134},     the  observer of system \dref{202010101753} can be designed as
\begin{equation}\label{202010101756}
\left\{\begin{array}{l}
\disp  \hat{w}_{t}(x,t)= \Delta \hat{w} (x,t)+\mu \hat{w}(x,t)+ K[C_vv(\cdot,t)-C_v\hat{v}(\cdot,t)]
 ,\ \   x\in \Omega, \crr
 \hat{w} (x,t)=0,\ \ \ x\in\Gamma_0,\ \ \disp \frac{\partial \hat{w}(x,t)}{\partial \nu}=u(x,t),\ \ x\in\Gamma_1,\crr
  \disp  \hat{v}_t(\cdot,t)=-\beta\hat{v}(\cdot,t)+Q B^*\hat{w}(\cdot,t) - L[C_vv(\cdot,t)-C_v\hat{v}(\cdot,t)]  \ \ \mbox{in}\ \ \Gamma_1,
  \end{array}\right.
\end{equation}
where $K$ and $L$ are tuning parameters  that can be chosen by the following scheme:
\begin{itemize}

   \item  Solve
 the following   equation
 \begin{equation} \label{202010112016}
 \beta P+P(A+\mu) + Q B^*=0
\end{equation}
to get  $P\in \mathcal{L}(L^2(\Omega),L^2(\Gamma_1))$;

\item Find $K$  to detect    system   $(A+\mu, C_vP)$;

  \item  Let $L=PK$.

\end{itemize}

By a straightforward computation, the solution of \dref{202010112016} is found to be
 \begin{equation} \label{20201011934}
 P= -Q B^* (\beta+\mu+A)^{-1}\in \mathcal{L}(L^2(\Omega),L^2(\Gamma_1)).
\end{equation}
By \dref{2020910949}, \dref{202010111134},  \dref{20201011931} and \dref{202010111141}, we have
 \begin{equation} \label{20201012835}
 C_vP= J_{q}^{\gamma}\in \mathcal{L}(L^2(\Omega),\R)\ \ \mbox{with}\ \ \gamma=-\beta-\mu.
\end{equation}
By Lemma \ref{lm202010111744}, \dref{202010111734}
and  \dref{2020871455},
 the operator $K$ can be chosen by   \dref{202010111744},
where $ (k_1,k_2,\cdots,k_N)^{\top}$ is a vector such that
  $\Lambda_N+(k_1,k_2,\cdots,k_N)^{\top}J_N$ is Hurwitz. As a result of \dref{202010111744}, \dref{20201011931} and \dref{20201011934},
   \begin{equation}\label{202010121035}
   L=PK=\sum_{j=1}^{N}k_jP\phi_j   =     -  \sum_{j=1}^{N}k_jQ B^* (\beta+\mu+A)^{-1}\phi_j
   =   -  \sum_{j=1}^{N}k_jq(x)  \xi_j (x), \ \ x\in\Gamma_1,
   \end{equation}
where
 \begin{equation} \label{20201011947}
 \left\{\begin{array}{l}
\disp (\beta+\mu+\Delta) \xi_j=\phi_j\ \ \mbox{in}\ \ \Omega,\crr
\disp  \xi_j(x)=0,\ x\in\Gamma_0,\ \ \frac{\partial \xi_j(x)}{\partial \nu}=0,\ x\in \Gamma _1.
\end{array}\right.j=1,2,\cdots,N.
\end{equation}
  Combining \dref{202010111744} and \dref{202010121035}, the observer \dref{202010101756}  turns to be
 \begin{equation}\label{202010121033}
\left\{\begin{array}{l}
\disp  \hat{w}_{t}(x,t)= \Delta \hat{w} (x,t)+\mu \hat{w}(x,t)+  [C_vv(\cdot,t)-C_v\hat{v}(\cdot,t)] \sum_{j=1}^{N}k_j\phi_j(x)
 ,\ \   x\in \Omega,  \crr
 \hat{w} (x,t)=0,\ \ \ x\in\Gamma_0,\ \ \disp \frac{\partial \hat{w}(x,t)}{\partial \nu}=u(x,t),\ \ x\in\Gamma_1,\crr
  \disp  \hat{v}_t(x,t)=-\beta\hat{v}(x,t)+Q B^*\hat{w}(x,t)  + \sum_{j=1}^{N}k_jq(x)  \xi_j (x)[C_vv(\cdot,t)-C_v\hat{v}(\cdot,t)] ,\  \ x\in \Gamma_1,
  \end{array}\right.
\end{equation}
 where
 $\xi_j$ is given by \dref{20201011947}, $j=1,2,\cdots,N$.
 By  \dref{202010111744}
  and \dref{202010121035},
 the observer can be written as the abstract form:
 \begin{equation}\label{202010121033abstract}
 \frac{d}{dt}(\hat{w} (\cdot,t),\hat{v} (\cdot,t))^{\top}=  \A (\hat{w} (\cdot,t),\hat{v} (\cdot,t))^{\top}+  (K,-L)  ^{\top}C_vv(\cdot,t) ,
\end{equation}
where the operator $\A: D(\A)\subset L^2(\Omega)\times L^2(\Gamma_1)\to  L^2(\Omega)\times L^2(\Gamma_1)$ is defined by
\begin{equation}\label{202010121441}
 \A=\begin{pmatrix}
A+\mu &-KC_v\\
QB^*&LC_v-\beta
    \end{pmatrix} \ \ \mbox{with}\ \  D(\A)=  D(A)\times L^2(\Gamma_1).
 \end{equation}

\begin{theorem}\label{Th20171226744}
  Suppose that the operator $A $ is  given by \dref{20209101543},
  $B^*$ is given by \dref{2020910949},  the eigenpairs   $\{(\phi_j(\cdot), \lambda_j)\}_{j=1}^{\infty}$   are  given by   \dref{20209101524}
 and  $Q$ is given by  \dref{20201011931}
with   $q(\cdot) \in L^2(\Gamma_1)$  satisfying \dref{20201091629q}.
Let the  integer $N$  satisfy   \dref{201912301959} and $\mu,\beta>0 $ satisfy
   \begin{equation}\label{202010121100}
-\beta-\mu \neq \lambda_j ,\ \ j=1,2,\cdots, N.
\end{equation}
  Then,
  for any      $({w}(\cdot,0), {v}(\cdot,0),\hat{w}(\cdot,0),\hat{v}(\cdot,0))^{\top}\in [L^2(\Omega)\times L^2(\Gamma_1)]^2$ and $u\in L^2_{\rm loc}([0,\infty);L^2(\Gamma_1))$,
  the observer \dref{202010121033} of system \dref{202010101654}
admits a unique   solution
 $(\hat{w} ,\hat{v} )^{\top}\in C([0,\infty);L^2(\Omega)\times L^2(\Gamma_1))$ such that
  \begin{equation} \label{20171226746}
 e^{\omega t}\|(w(\cdot,t)-\hat{w}(\cdot,t),v(\cdot,t)-\hat{v}(\cdot,t))\|_{L^2(\Omega)\times L^2(\Gamma_1)}\to0\ \ \mbox{as}\ \ t\to\infty,
\end{equation}
where $\omega$ is a positive constant that is independent of $t$.
\end{theorem}
\begin{proof}
 For any $({w}(\cdot,0), {v}(\cdot,0)  )^{\top}\in  L^2(\Omega)\times L^2(\Gamma_1) $ and $u\in L^2_{\rm loc}([0,\infty);L^2(\Gamma_1))$, it is well known that the control plant
\dref{202010101654} admits a unique solution $( {w} , {v} )^{\top}\in C([0,\infty);L^2(\Omega)\times L^2(\Gamma_1))$ such that $y_v\in L^2_{\rm loc}[0,\infty)$.
 Let
\begin{equation}\label{202010111958}
\left\{\begin{array}{l}
  \disp \tilde{w}(x,t)=w(x,t)-\hat{w} (x,t), \ \ x\in \Omega,\crr
\disp   \tilde{v}(s,t)=v(s,t)-\hat{v}(s,t),\ \ s\in \Gamma_1,
  \end{array}\right.\ \ t\geq0.
\end{equation}
Then, the errors are governed by
 \begin{equation}\label{202010121419}
\left\{\begin{array}{l}
\disp  \tilde{w}_{t}(x,t)= \Delta \tilde{w} (x,t)+\mu \tilde{w}(x,t)- C_v\tilde{v}(\cdot,t)  \sum_{j=1}^{N}k_j\phi_j(x)
 ,\ \   x\in \Omega,  \crr
 \tilde{w} (x,t)=0,\ \ \ x\in\Gamma_0,\ \ \disp \frac{\partial \tilde{w}(x,t)}{\partial \nu}=0,\ \ x\in\Gamma_1,\crr
  \disp  \tilde{v}_t(x,t)=-\beta\tilde{v}(x,t)+Q B^*\tilde{w}(x,t)  - C_v\tilde{v}(\cdot,t) \sum_{j=1}^{N}k_jq(x)  \xi_j (x) ,\  \ x\in \Gamma_1.
  \end{array}\right.
\end{equation}
By \dref{202010121441},  \dref{202010121035} and \dref{202010111744}, system \dref{202010121419} can be written abstractly
\begin{equation}\label{202010112001}
\frac{d}{dt}(\tilde{w} (\cdot,t),\tilde{v} (\cdot,t))^{\top}=  \A (\tilde{w} (\cdot,t),\tilde{v} (\cdot,t))^{\top}.
\end{equation}
Inspired by \cite{FPart2}, we introduce the following transformation
   \begin{equation}\label{2020109938obser}
\mathbb{P}(f,g)^{\top}=(f,g+Pf)^{\top},\ \ (f,g)^{\top}\in L^2(\Omega)\times L^2(\Gamma_1),
\end{equation}
where $P\in \mathcal{L}(L^2(\Omega),L^2(\Gamma_1))$  is the solution of system \dref{202010112016}.
 Then
$\mathbb{P}$ is invertible and its inverse is given by
 \begin{equation}\label{202010112014}
\mathbb{P}^{-1}(f,g)^{\top}=(f,g-Pf)^{\top},\ \ (f,g)^{\top}\in L^2(\Omega)\times L^2(\Gamma_1).
\end{equation}
Moreover, a simple computation shows that (see, e.g., \cite[Theorem 6.1]{FPart2})
 \begin{equation}\label{202010121457}
 \mathbb{P}\A \mathbb{P}^{-1}=\A_\mathbb{P}, \ \ D(\A_\mathbb{P})= \mathbb{P}D(\A),
\end{equation}
where
\begin{equation}\label{202010121459}
\disp \A_{\mathbb{P}} =\begin{pmatrix}
{A}+\mu+KC_vP   &-KC_v\\
0& -\beta
    \end{pmatrix}\ \ \mbox{with}\ \
\disp D(\A_{\mathbb{P}})= D(A)\times  L^2(\Gamma_1)  .
 \end{equation}
By Lemma \ref{lm202010111744} and \dref{20201012835}, the operator  ${A}+\mu+KC_vP$
generates an exponentially stable $C_0$-semigroup
 on $L^2(\Omega)$.
 Thanks to the   block-triangle structure and \cite{FPart1}, the operator $\A_{\mathbb{P}}$
 generates an exponentially  stable  $C_0$-semigroup $e^{\A_{\mathbb{P}}t}$
 on $L^2(\Omega)\times L^2(\Gamma_1)$. By virtue of the similarity \dref{202010121457},
 the operator $\A $ also
 generates an exponentially  stable $C_0$-semigroup $e^{\A t}$
 on $L^2(\Omega)\times L^2(\Gamma_1)$. As a result, the error system with initial state
  $(\tilde{w}(\cdot,0),\tilde{v}(\cdot,0) )^{\top}=(w (\cdot,0)-\hat{w}(\cdot,0),v (\cdot,0)-\hat{v}(\cdot,0))^{\top}\in L^2(\Omega)\times L^2(\Gamma_1)$
  admits a unique solution
 $(\tilde{w} ,\tilde{v} )^{\top}\in C([0,\infty);L^2(\Omega)\times L^2(\Gamma_1))$
 such that
  \begin{equation} \label{20171226746error}
 e^{\omega t}\|(\tilde{w}(\cdot,t),\tilde{v}(\cdot,t))\|_{L^2(\Omega)\times L^2(\Gamma_1)}\to0\ \ \mbox{as}\ \ t\to\infty,
\end{equation}
where $\omega$ is a positive constant that is independent of $t$.
Let
 \begin{equation}\label{202010112018}
  (\hat{w}(\cdot,t),\hat{v}(\cdot,t)) = (
  w (\cdot,t)-\tilde{w}(\cdot,t),v (\cdot,t)-\tilde{v}(\cdot,t)).
\end{equation}
Then, a straightforward computation shows that such a defined    $(\hat{w} ,\hat{v} )^{\top}\in C([0,\infty);L^2(\Omega)\times L^2(\Gamma_1))$  is a solution of system
\dref{202010121033abstract} or equivalently, system \dref{202010121033}. Moreover, \dref{20171226746} holds due to \dref{20171226746error} and \dref{202010111958}.
Owing to the linearity of system  \dref{202010121033}, the solution is unique.
 \end{proof}


\section{Numerical simulations}\label{Numerical}

 In this section,
we present some numerical simulations for the closed-loop
system \dref{20201092100} to demonstrate the  theoretical results visually.
In order to avoid the difficulty of
numerical discretization,
  we consider the unstable heat   system  in the  rectangular domain
   $\Omega=\left \{(x ,y )\in \mathbb{R}^2\ |\  0< x <1, 0< y  <1\right\}  $.
  The actuator is installed on the boundary
  $$\Gamma_1=\left\{(x ,y )\in \mathbb{R}^2\ |\  x =1,0\leq  y \leq1\right \} \cup   \left\{(x ,y )\in \mathbb{R}^2\ |\   y =1,0\leq x \leq1  \right\}.$$
    The fixed boundary is
   $$\Gamma_0=\left\{(x ,y )\in \mathbb{R}^2\ |\  x =0,0\leq  y <1\right \} \cup   \left\{(x ,y )\in \mathbb{R}^2\ |\   y =0,0\leq x <1  \right\}.$$
%

  We
adopt the finite difference scheme  to discretize system
\dref{20201092100} directly.  The numerical
results are programmed in Matlab.
Inspired by \cite{ZuaZuasimulation} where
the  uniform exponential decay with respect to the mesh
size is obtained by   the finite difference method,
  the space step $h$ and time
 step  $\tau$ are   taken as $h =\tau=   0.05$.
 The initial state and tuning parameters  are  chosen as
\begin{equation}\label{20198101148}\left\{\begin{array}{l}
   \disp w(x ,y ,0)= x \sin2\pi y   ,\ \ v(x ,y ,0)=0, \crr
   {p}(x ,y )=\sin x  \sin y ,\ \ \mu=6,\ \alpha=3,\ \ N=1,\ \ l_1= 15.
\end{array}\right.
\end{equation}
 By a simple numerical computation, the largest eigenvalue of the operator \dref{20209101543} on  $\Omega$ is $\lambda_1\approx -4.6947$. This implies that  $\lambda_1+\mu>0$ and hence the open-loop system \dref{2020981903} is unstable.

The   the initial state and the  final state  of closed-loop system
 \dref{20201092100}  are  plotted in  Figure \ref{Fig1}.
 \begin{figure}[h]
\hspace{0.5cm}\subfigure[   $w(x,y,0)$.]
  {\includegraphics[width=0.4\textwidth]{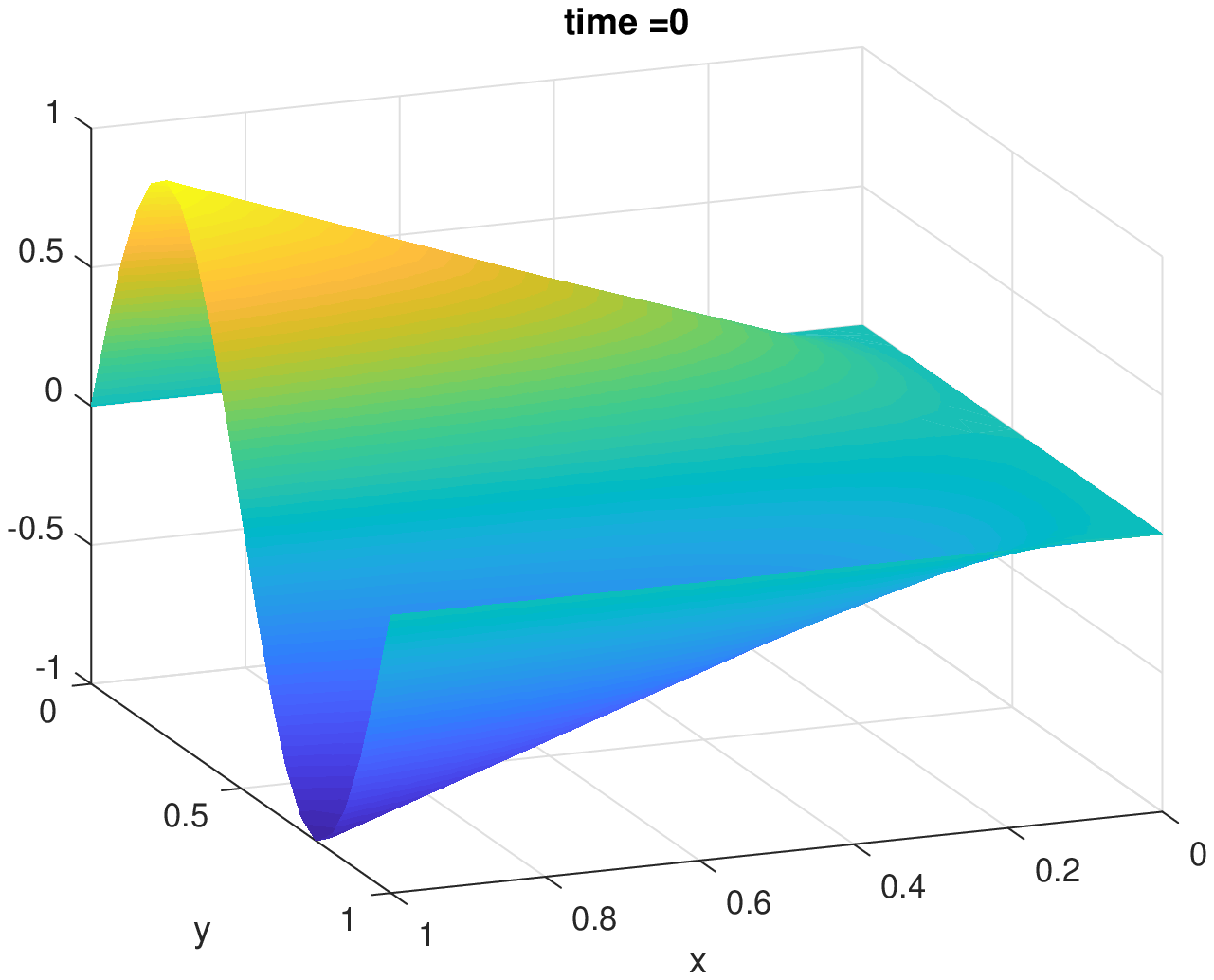}}
 \hspace{2cm}\subfigure[  $w(x,y,4)$.]
  {\includegraphics[width=0.4\textwidth]{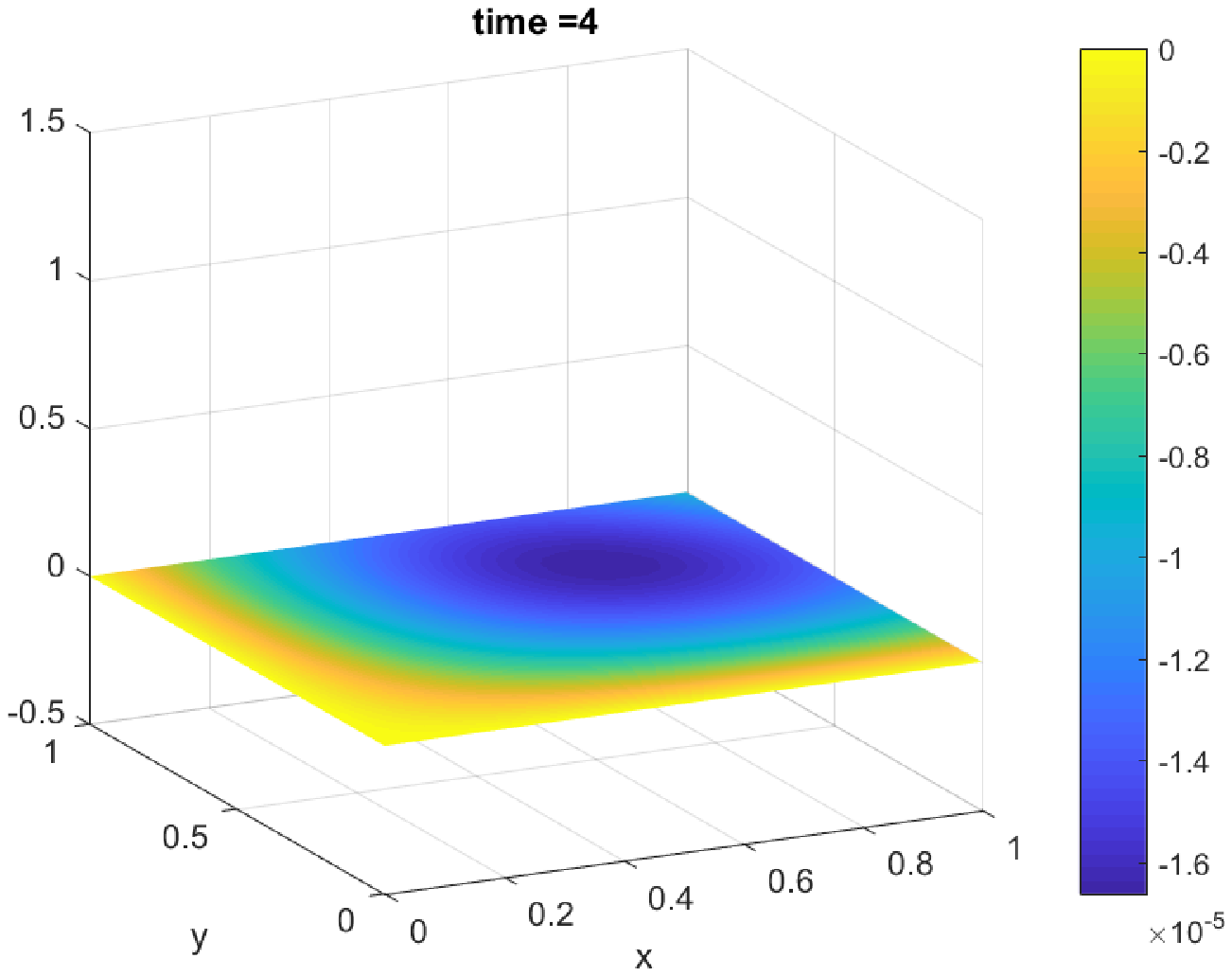}}
         \caption{  The initial state and the final state. }\label{Fig1}
          \end{figure}
          \begin{figure}[h]
\hspace{0.5cm}\subfigure[   $w(x,0.5,t)$ with control.]
  {\includegraphics[width=0.4\textwidth]{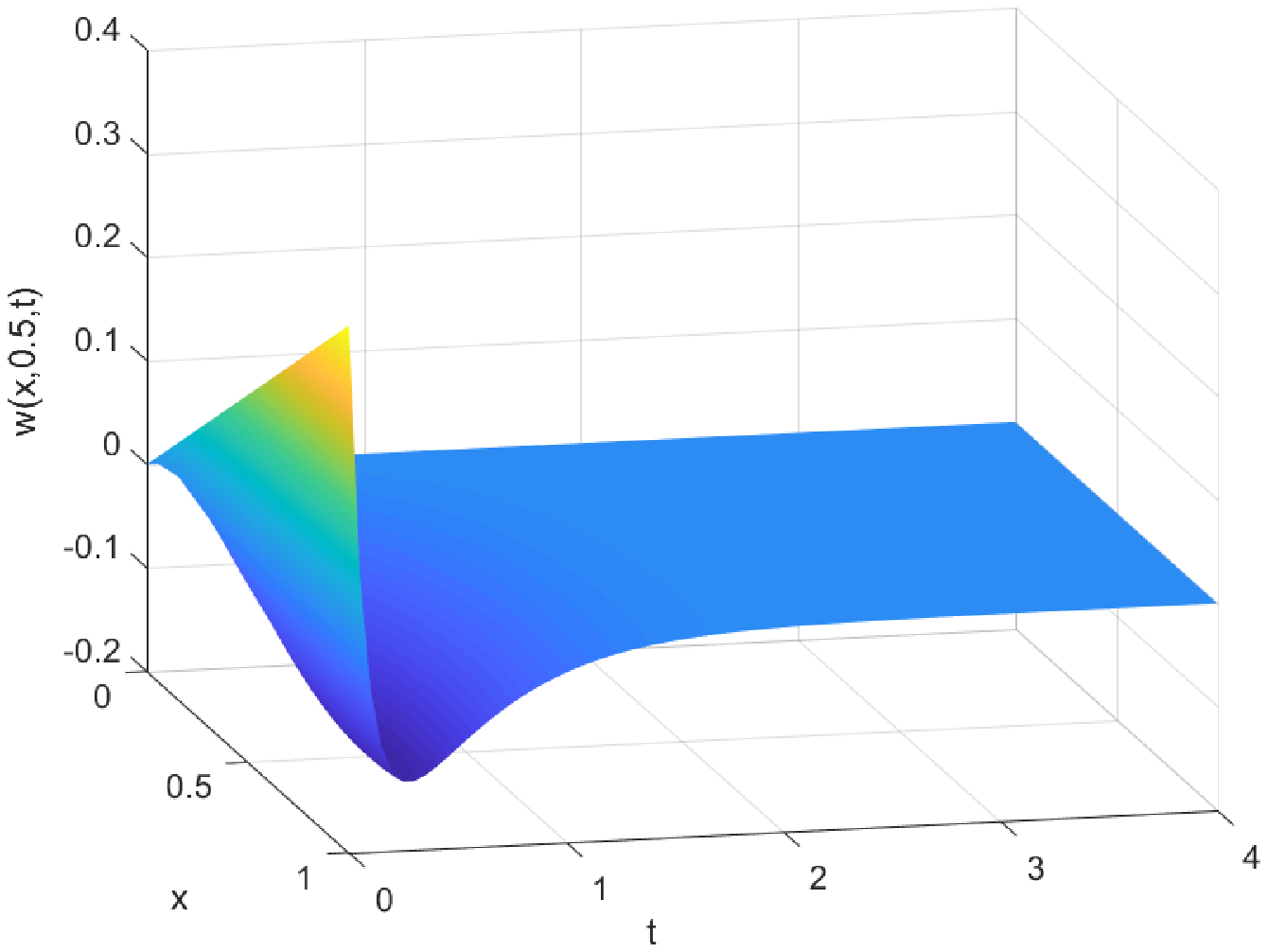}}
 \hspace{2cm}\subfigure[  $w(x,0.5,t)$ without control.]
  {\includegraphics[width=0.4\textwidth]{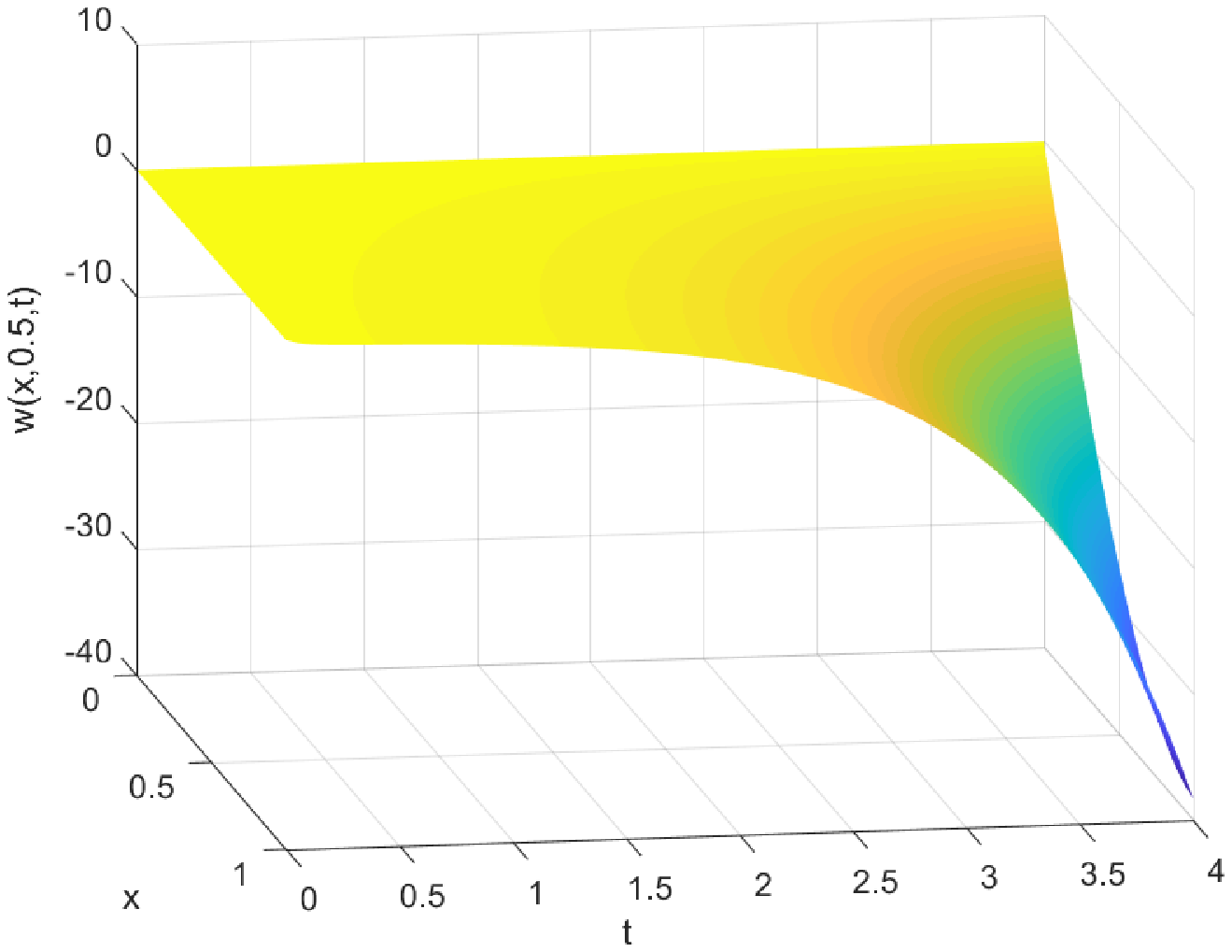}}
         \caption{  State trace    with control and state trace without control. }\label{Fig2}
          \end{figure}
In order to demonstrate the dynamic evolution of the closed-loop system, the state
trace $w(x,0.5,t)$
 is plotted in Figure~\ref{Fig2}(a).  The same  state
trace without control
 is plotted in Figure~\ref{Fig2}(b) for comparison.
 The  distributed
   control traces   $v(x,y,t )$ are plotted in Figure~\ref{Fig3}(a) and \ref{Fig3}(b).
   To demonstrate the decay rate,
   the  logarithmic state norm $\|w(\cdot,t)\|_{L^2(\Omega)}$ decay curve
  and the    curve of the   state norm itself are plotted in
  Figure~\ref{Fig4}(a) and  Figure~\ref{Fig4}(b), respectively.
 \begin{figure}[h]
\hspace{0.5cm}\subfigure[   $v(x,1,t)$.]
  {\includegraphics[width=0.4\textwidth]{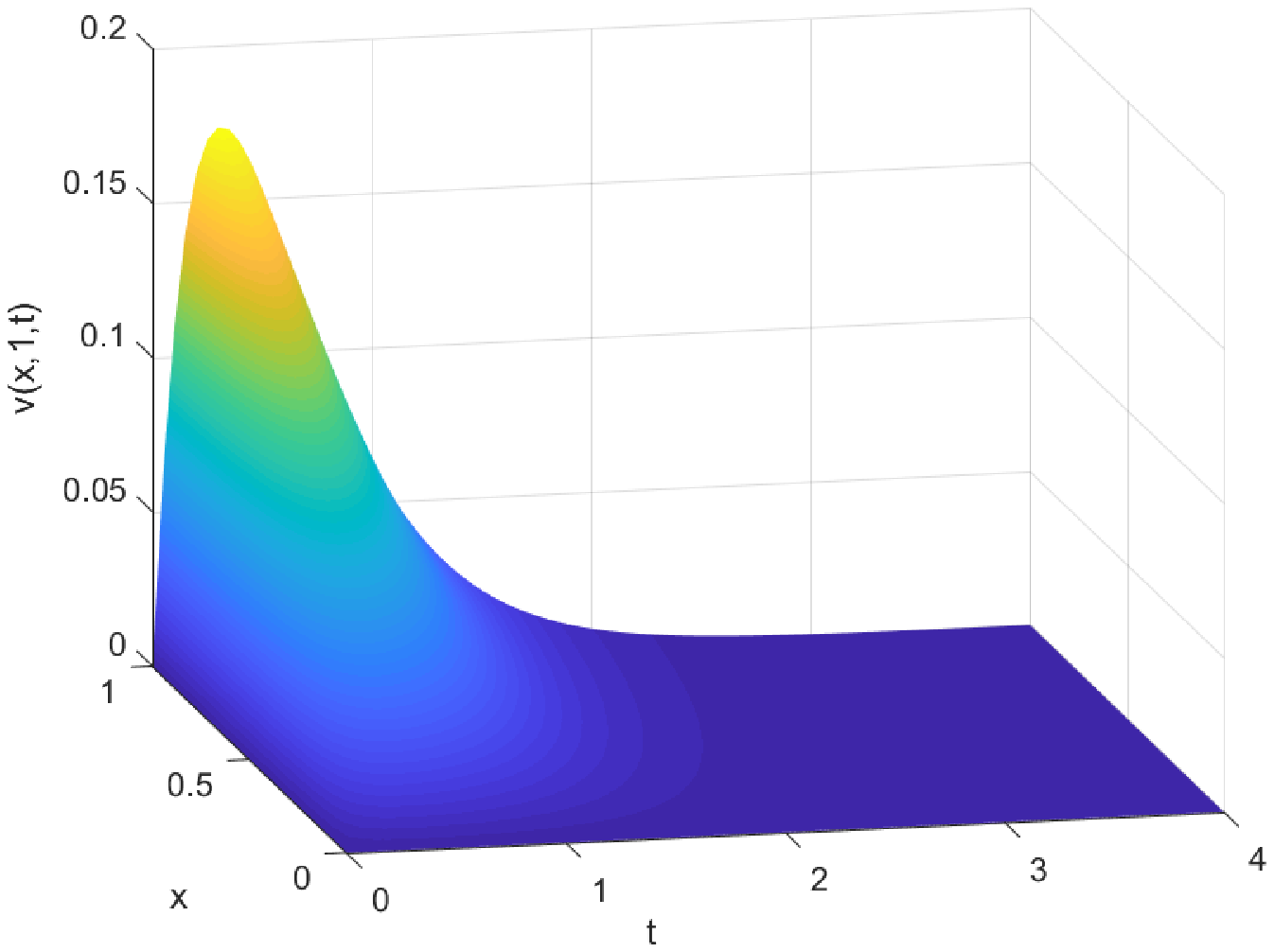}}
 \hspace{2cm}\subfigure[  $v(1,y,t)$.]
  {\includegraphics[width=0.4\textwidth]{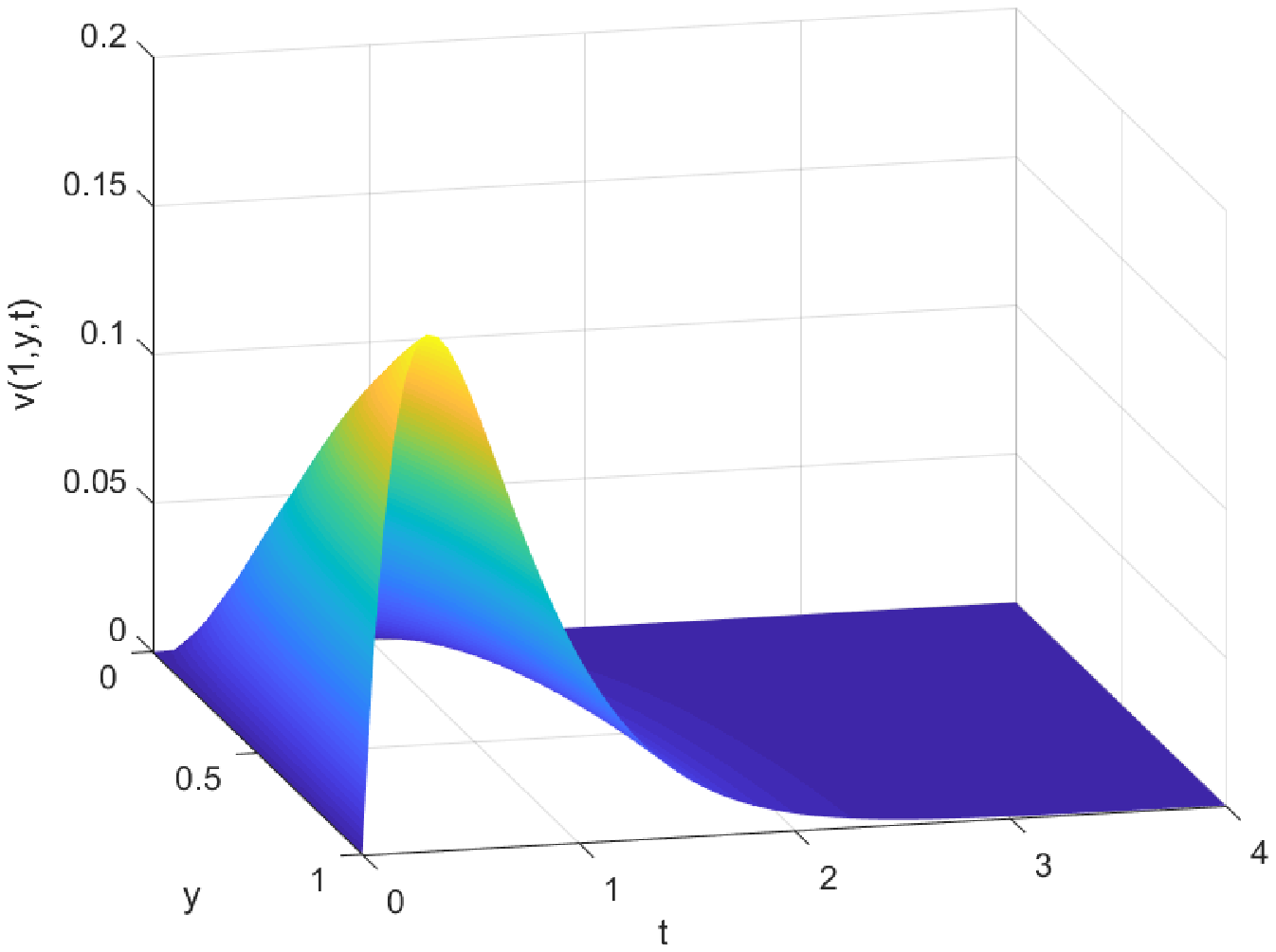}}
         \caption{  Controller  traces. }\label{Fig3}
          \end{figure}
          \begin{figure}[h]
\hspace{0.5cm}\subfigure[   $\log(\|w(\cdot,t)\|_{L^2(\Omega)})$.]
  {\includegraphics[width=0.4\textwidth]{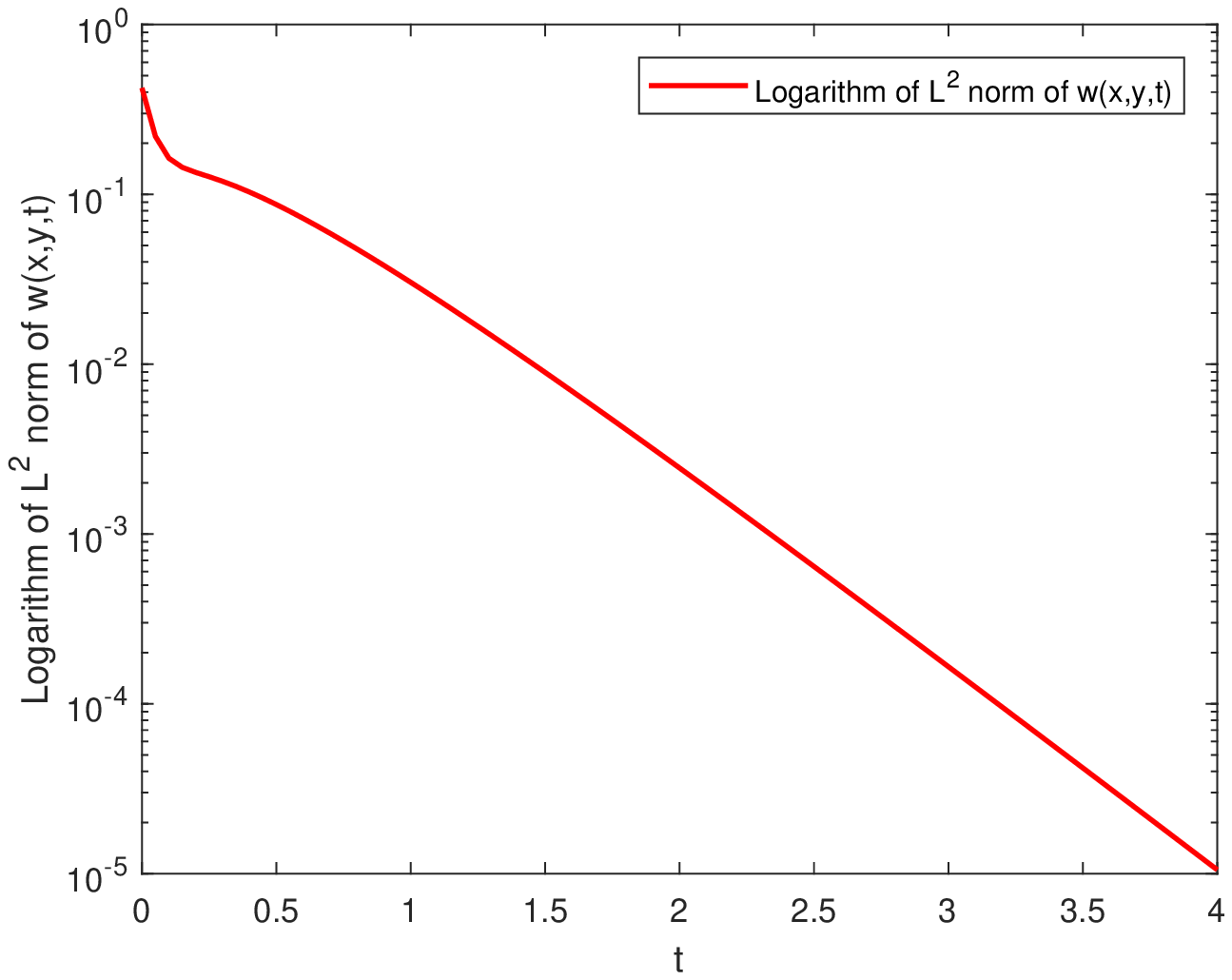}}
 \hspace{2cm}\subfigure[  $\|w(\cdot,t)\|_{L^2(\Omega)}$.]
  {\includegraphics[width=0.4\textwidth]{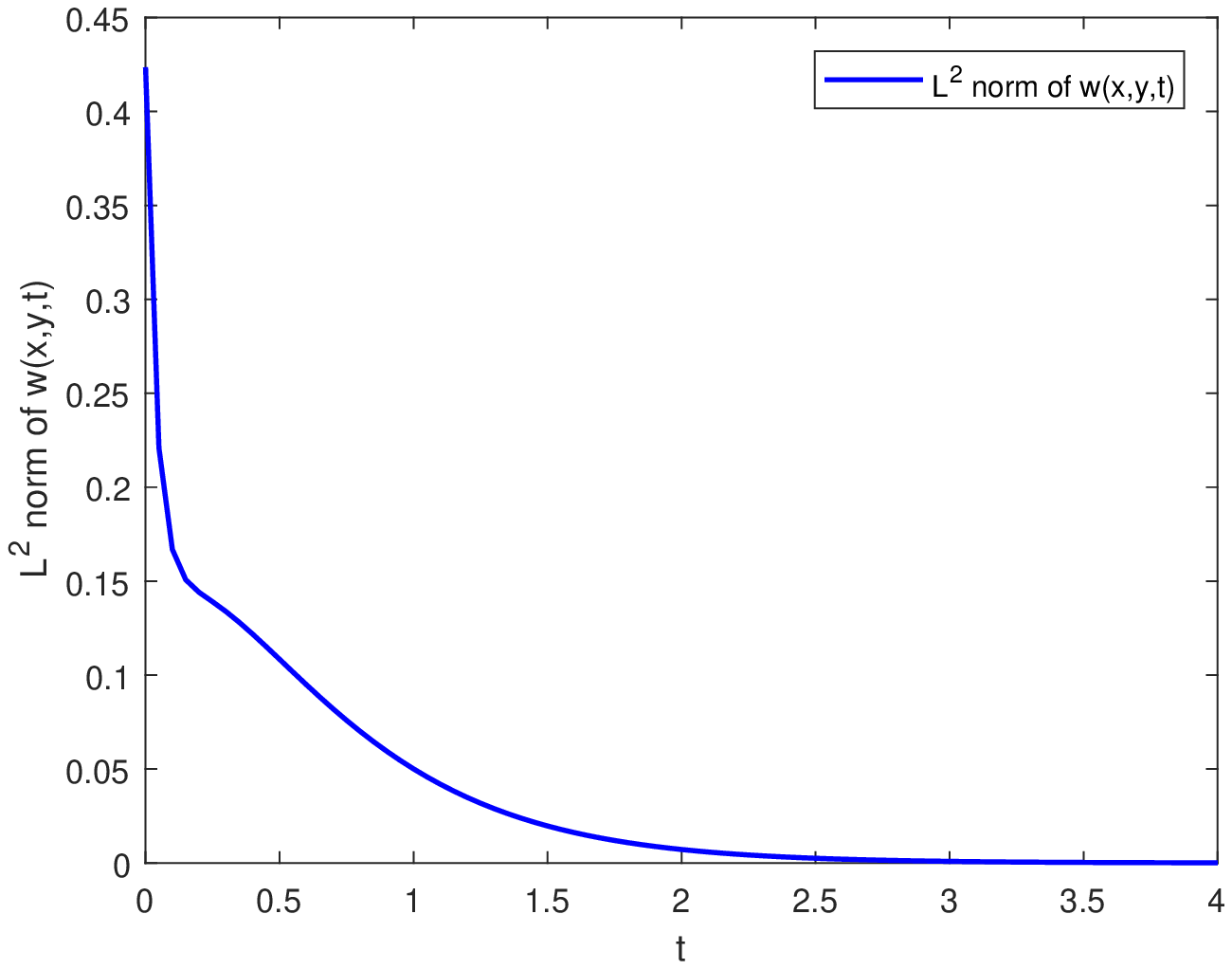}}
         \caption{  Decays of the state norm. }\label{Fig4}
          \end{figure}
 From these Figures~\ref{Fig1}-\ref{Fig4} we observe that the state of the control plant are
stabilized effectively despite the presence of the unstable  source term $\mu w(x,y,t)$.
Moreover,   the dynamic evolution is
smooth.    Figure \ref{Fig4}  implies that  the state norm  decays to zero  exponentially.
So all the   convergence in the closed-loop system   is very fast.

\section{Conclusions}\label{Conclusions}
In this paper, we consider the stabilization and observation for the unstable  heat equation in a general multi-dimensional domain. The newly developed dynamics compensation approach  and the    finite-dimensional spectral truncation technique are  exploited to
treat the
difficulties caused by instability.
Both the full state feedback law and the state observer are designed. The closed-loop system and the
observation error are convergent   to zero exponentially as $t\to \infty $.
 The developed method in this paper
provides  a new choice, in addition to the PDE backstepping method,  for dealing  with unstable PDEs,
especially for multi-dimensional unstable PDEs.
It is very interesting to extend  this new method to other unstable or anti-stable PDEs such as
the multi-dimensional wave equation and
 Euler-Bernoulli beam equation, which are our future works.


\section{Appendix}

\begin{lemma}\label{Lm202010101612}
  For any positive integer $N$, there exists a  function  $p\in L^2(\Gamma_1)$ such that
  \dref{20201091629} holds.
 \end{lemma}
\begin{proof}
Choose $p_1(x)=\phi_1(x)$ for any $x\in \Gamma_1$. Then
$\langle p_1,\phi_1\rangle_{L^2(\Gamma_1)} \neq0$. If
$ \langle \phi_2,\phi_1\rangle_{L^2(\Gamma_1)}\neq0$, then let $p_2=p_1$. Otherwise,
let
 \begin{equation} \label{20201010841}
 p_2(x)=p_1(x)+\phi_2(x).
\end{equation}
Then, $\langle p_2,\phi_1\rangle_{L^2(\Gamma_1)}\neq0$ and $\langle p_2,\phi_2\rangle_{L^2(\Gamma_1)}\neq0$.
Suppose that  we have obtain $p_{N-1}$ such that
 \begin{equation} \label{202010101602}
\langle p_{N-1}, \phi_j\rangle_{L^2(\Gamma_1)}\neq0,\ \ j=1,2,3,\cdots,N-1.
 \end{equation}
 If  $\langle p_{N-1}, \phi_N\rangle_{L^2(\Gamma_1)}\neq0$, we choose $p_N=p_{N-1}$. Otherwise,
  \begin{equation} \label{202010101604}
 p_N(x)=p_{N-1}(x)+\gamma \phi_N(x),
 \end{equation}
 where $\gamma$ small enough such that
 \begin{equation} \label{202010101605}
 \langle p_{N-1},\phi_j\rangle_{L^2(\Gamma_1)}+\gamma \langle\phi_N,\phi_j\rangle_{L^2(\Gamma_1)} \neq0,\ \ j=1,2,\cdots,N.
 \end{equation}
Therefore, the proof is complete due to the   mathematical induction.
\end{proof}

\begin{lemma}\label{Lm202010131720}
For any positive integer $N$,  define
\begin{equation} \label{202010131721}
\Lambda_N={\rm diag}(\lambda_1,\lambda_2,\cdots,\lambda_N)
\end{equation}
and
\begin{equation} \label{202010131721B}
B_N= (b_1,b_2,\cdots,b_N)^{\top},
\end{equation}
where $b_k\neq0$, $k=1,2,\cdots,N$  and
\begin{equation} \label{202010131723}
 \lambda_i \neq \lambda_j, \ \ i\neq j,\ \ i,j=1,2,\cdots,N.
\end{equation}
Then, system $(\Lambda_N, B_N)$ is controllable.

 \end{lemma}
\begin{proof}
  By a simple computation, the controllability matrix  of system  $(\Lambda_N, B_N)$  is
  \begin{equation} \label{202010131736}
P_c=\begin{pmatrix}
 b_1&\lambda_1b_1&\cdots&\lambda_1^{N-1}b_1 \\
 b_2&\lambda_2b_2&\cdots&\lambda_2^{N-1}b_2 \\
 \vdots&\vdots&\cdots&\vdots\\
 b_N&\lambda_Nb_N&\cdots&\lambda_N^{N-1}b_N\\
\end{pmatrix}.
\end{equation}
Furthermore,
\begin{equation} \label{202010131741}
| P_c|=b_1b_2\cdots b_N\left|\begin{pmatrix}
 1&\lambda_1 &\cdots&\lambda_1^{N-1}  \\
 1&\lambda_2 &\cdots&\lambda_2^{N-1}  \\
 \vdots&\vdots&\cdots&\vdots\\
 1&\lambda_N &\cdots&\lambda_N^{N-1} \\
\end{pmatrix}\right|=b_1b_2\cdots b_N \prod_{1\leq i<j\leq N}(\lambda_i-\lambda_j)\neq0 .
\end{equation}
Therefore, the proof is  complete due to the Kalman's controllability rank condition.
\end{proof}
\end{document}